%% file: gandalf2011.tex
\documentclass[copyright,creativecommons]{eptcs}

\usepackage[english]{babel}
\usepackage[latin1]{inputenc}
\usepackage{amsmath,amssymb,amsfonts,mathrsfs,latexsym,stmaryrd}
\usepackage{array}
\usepackage{url}
\usepackage[newitem,newenum,neverdecrease]{paralist}
\usepackage{ifpdf}
\usepackage{color}
\definecolor{darkcyan}{rgb}{0,0.55,0.55}
\usepackage[OT1]{eulervm}
\usepackage{times}
\usepackage{graphicx}
\usepackage{latexsym}
\usepackage{fancybox}
\usepackage{tikz}
\usepgflibrary{decorations.pathreplacing} 
\usepgflibrary[decorations.pathreplacing] 
\usetikzlibrary{decorations.pathreplacing} 
\usetikzlibrary[decorations.pathreplacing]
\usepackage{stackrel}
\usepackage{xspace}
\usepackage{amsthm}

\input{fixpseudocode}
\newtheorem{definition}{Definition}
\newtheorem{proposition}{Proposition}
\newtheorem{theorem}{Theorem}

\newcommand{\AP}{\mathcal{A}\mathcal{P}}
\newcommand{\kphi}{k_{\varphi}}

\newcommand{\oy}{\overline{y}}

\newcommand{\G}{[G]}
\newcommand{\Essentials}{\cE\cS}
\newcommand{\Blocked}{\cB locked}

\newcommand{\dl}{\lozenge_l}
\newcommand{\dr}{\lozenge_r}
\newcommand{\Dl}{\Box_l}
\newcommand{\Dr}{\Box_r}
\newcommand{\mmodels}{\Vdash}

\newcommand{\lenlt}[1]{\len{<}{#1}}
\newcommand{\leneq}[1]{\len{=}{#1}}

\newcommand{\len}[2]{\ensuremath{\mathsf{len_{\mathcal{#1} #2}}}}

\providecommand{\event}{GandALF 2011} 
\usepackage{breakurl}             
\input{generic}
\title{An Optimal Decision Procedure for MPNL over  the Integers}
\author{Davide Bresolin
\institute{Department of Computer Science \\ University of Verona (Italy)}\\
{\small davide.bresolin@univr.it}
\and
Angelo Montanari
\institute{Department of Mathematics and Computer Science \\ University of Udine (Italy)}\\
{\small angelo.montanari@uniud.it}
\and
Pietro Sala
\institute{Department of Computer Science \\ University of Verona (Italy)}\\
{\small pietro.sala@univr.it}
\and
Guido Sciavicco
\institute{University of Murcia (Spain) and \\ University for Information Science and Technology, \\
Ohrid (Macedonia)}\\
{\small guido@um.es}
}
\def\titlerunning{An Optimal Decision Procedure for MPNL over  the Integers}
\def\authorrunning{D. Bresolin et al.}
\begin{document}
\maketitle

\begin{abstract}
Interval temporal logics provide a natural framework for qualitative and
quantitative temporal reasoning over interval structures, where the truth
of formulae is defined over intervals rather than points.
In this paper, we study the complexity of the satisfiability problem for  Metric
Propositional Neighborhood Logic (MPNL). MPNL  features two modalities
to access  intervals ``to the left'' and ``to the right'' of the current 
one, respectively, plus an infinite set of length constraints.
MPNL, interpreted over the naturals, has been recently shown 
to be decidable by a doubly exponential procedure. We improve such a 
result by proving that MPNL is actually EXPSPACE-complete (even when 
length constraints are encoded in binary), when interpreted over finite 
structures, the naturals, and the integers, by developing 
an EXPSPACE decision procedure for MPNL over the integers, which can be 
easily tailored to finite linear orders and the naturals (EXPSPACE-hardness 
was already known).
\end{abstract}

\input{Introduction}

\section{The logic MPNL}\label{sec:logic}
\input{Logic}

\medskip

\section{Atoms, types, dependencies, and compass structures}\label{sec:types}
\input{Compass}
\input{Decidabilityfinite}

\input{Decidabilitynatural}
\input{Decidabilityinteger}
\input{Decisionprocedure}
\section*{Acknowledgements}
We would like to thank the anonymous reviewers for their useful comments
and suggestions. 
This research has been partially supported by the EU project \emph{FP7-ICT-223844 CON4COORD}
(Davide Bresolin), the Italian PRIN project \emph{Innovative and multi-disciplinary 
approaches for constraint and preference reasoning} (Angelo Montanari and Pietro Sala), 
and the Spanish MEC project \emph{TIN2009-14372-C03-01} (Guido Sciavicco).

\bibliographystyle{eptcs}
\bibliography{gandalf2011}
\end{document}

%% file: fixpseudocode.tex
\usepackage{pseudocode}
\usepackage{chngcntr}

\counterwithout{pseudocode}{section}

\newcounter{newpseudonum}[pseudocode]

\ifpdf
  \providecommand{\refline}[1]{\hyperref[#1]{(\ref*{#1})}}
\else
  \providecommand{\refline}[1]{\ref*{#1}}
\fi

\renewcommand{\ELSE}{\\\pcodetab{1}\mbox{ \bfseries \makebox[0pt][l]{else}\phantom{then} }}
\renewcommand{\RETURN}[1]{\ifthenelse{\equal{#1}{} }{\mbox{\bfseries return}}{\mbox{\bfseries return}#1}}
\newcommand{\FUNCTION}[2]{\mbox{\bfseries proc }\mbox{\textsc{#1}}\left(\ensuremath{#2}\right)\\}

\newlength{\pcodewidth}
\newenvironment{code}[1]{
\begin{Sbox}
\!\!\begin{minipage}{#1}
\bfseries
\noindent
\scriptsize
$$
\begin{array}{@{\hspace*{1ex}}lr@{}}
}{
\end{array}
$$
\end{minipage}\vspace{-2mm}
\end{Sbox}
\shadowbox{\TheSbox}{}
}

%% file: generic.tex
\ifpdf
\else
\fi

\newenvironment{dotlist}{\begin{compactitem}}{\end{compactitem}}

\newcommand{\refcond}[1]{\hyperref[#1]{C\ref*{#1}}}
\newcommand{\refprop}[1]{\hyperref[#1]{P\ref*{#1}}}
\newcommand{\refaxiom}[1]{\hyperref[#1]{A\ref*{#1}}}

\newcommand{\bbD}{\mathbb{D}}

\newcommand{\bbN}{\mathbb{N}}

\newcommand{\bbZ}{\mathbb{Z}}
\newcommand{\bbP}{\mathbb{P}}

\newcommand{\cA}{\mathcal{A}}
\newcommand{\cB}{\mathcal{B}}
\newcommand{\cC}{\mathcal{C}}

\newcommand{\cE}{\mathcal{E}}

\newcommand{\cG}{\mathcal{G}}

\newcommand{\cK}{\mathcal{K}}
\newcommand{\cL}{\mathcal{L}}
\newcommand{\cM}{\mathcal{M}}

\newcommand{\cO}{\mathcal{O}}

\newcommand{\cR}{\mathcal{R}}
\newcommand{\cS}{\mathcal{S}}
\newcommand{\cT}{\mathcal{T}}

\newcommand{\s}[1]{\vspace{#1mm}}

\providecommand{\mit}{\mathit}
\renewcommand{\mit}{\mathit}

\makeatletter
\def\shortrightarrowfill@{\arrowfill@\relbar\relbar\shortrightarrow}
\newcommand{\ort}{\mathpalette{\overarrow@\shortrightarrowfill@}}
\makeatother
\makeatletter
\def\shortleftarrowfill@{\arrowfill@\relbar\relbar\shortleftarrow}
\newcommand{\olft}{\mathpalette{\overarrow@\shortleftarrowfill@}}
\makeatother
\makeatletter
\def\shortleftrightarrowfill@{\arrowfill@\relbar\relbar\leftrightarrow}
\newcommand{\olftrt}{\mathpalette{\overarrow@\shortleftrightarrowfill@}}
\makeatother

\newcommand{\nin}{\not\in}
\newcommand{\sat}{\vDash}

\newcommand{\vel}{\;\vee\;}

\newcommand{\lenf}[1]{{\left| #1 \right|}}

\renewcommand{\tt}[3]{\text{\small $\biggl(\begin{smallmatrix} #1 \\ #2 \\ #3 \end{smallmatrix}\biggr)$}}

\newtheorem{lemma}{Lemma}

\newcommand{\closure}{\cC\mit{l}}

\newcommand{\type}{\cT\mit{ype}}
\newcommand{\req}{\cR\mit{eq}}
\newcommand{\obs}{\cO\mit{bs}}

\newcommand{\labeledrightarrow}[1]{\overset{\text{\raisebox{-0.1ex}[0ex][-0.1ex]{$_{#1\,}$}}}{\longrightarrow}}

\newcommand{\dep}[1]{\,\text{\raisebox{-0.2ex}{$\labeledrightarrow{#1}$}}\,}



%% file: Introduction.tex
\section{Introduction}

Interval temporal logics provide a natural framework for
temporal representation and reasoning on interval structures 
over linearly (or partially) ordered domains. They take time 
intervals as the primitive ontological entities and define 
truth of formulae with respect to them instead of to time 
instants. Modal operators of interval temporal logics
correspond to binary relations between pairs of intervals
(in fact, an interval temporal logic of ternary interval relations
was developed by Venema in \cite{chopping_intervals}).
In the realm of interval temporal logics, a prominent role
is accorded to Halpern and Shoham's modal logic of time 
intervals (HS), whose modalities make it possible to express 
all Allen's binary interval relations~\cite{interval_relations}.

Interval-based temporal formalisms have been extensively used in various 
areas of computer science and artificial intelligence, including
hardware specification and verification, constraint processing,
planning and plan validation, theories of action and change, and
natural language understanding. However, many applications impose
severe syntactic and semantic restrictions that considerably weaken 
their expressive power.
Interval temporal logics relax these restrictions, thus allowing 
one to express much more complex temporal properties. Unfortunately, 
most of them, including HS and the majority of its fragments,
turn out to be undecidable (a comprehensive survey on interval 
logics can be found in \cite{Goranko04}; an up-to-date
picture of decidability and undecidability results about them 
can be obtained from~\cite{thesisdario,thesispietro}). 


One of the few cases of a decidable temporal logic with genuine interval
semantics, that is, not reducible to point-based semantics, is the
propositional logic of temporal neighborhood (Propositional
Neighborhood Logic, PNL for short), interpreted over various classes 
of temporal structures, including all, dense, discrete, and finite 
linear orders, as well as rational, integer, and natural 
numbers~\cite{Goranko03a}. 
PNL is the fragment of HS featuring two modalities corresponding to 
Allen's relations {\em meets} and {\em met by} (the one is the inverse
of the other). Decidability of PNL with respect to various classes of
linear orders has been proved in \cite{Bresolin09b} via a reduction to 
the satisfiability problem for the two-variable fragment of first-order 
logic for binary relational structures over ordered domains \cite{otto:2001a}. 
Decidability of PNL with respect to other classes of linear orders via
a direct model-theoretic argument has been recently shown in \cite{rr01}, 
where tableau-based optimal decision procedures for PNL, interpreted 
in the considered classes of linear orders, have also been developed.

Despite its seeming simplicity, PNL is well-suited for a number of 
concrete application domains. One of them is that of transaction-time 
databases (also called append-only databases), that keep track of 
the sequence of timestamped versions of the database, where 
information is never removed and new information is appended 
to existing information, respecting the temporal ordering. However, 
in such an application domain as well in various other ones, 
a metric dimension turns out to be a very useful ingredient.
A metric extension of PNL has been developed by Bresolin et al.\
in~\cite{sosym}. The resulting interval temporal logic, called Metric 
PNL (MPNL for short), pairs PNL modalities with a family of special 
atomic propositions expressing integer constraints (equalities and 
inequalities) on the length of the intervals over which they are 
evaluated. The authors show that the satisfiability 
problem for MPNL, interpreted over natural numbers, is decidable.
However, they leave the precise characterization of its complexity as 
an open problem. Metric constraints in MPNL are expressed 
in terms of some $k \in \mathbb N$. When $k$ is a constant of the 
formula or it is expressed in unary, MPNL is NEXPTIME-complete, but 
when $k$ is expressed in binary, then the satisfiability problem for 
MPNL has been shown to be somewhere in between EXPSPACE and 2NEXPTIME
only.


In this paper, we focus our attention on MPNL with a binary encoding of 
metric constraints. We first provide an original model-theoretic proof 
of the decidability of its satisfiability problem over finite 
linear orders, natural numbers, and integer numbers. As a
matter of fact, the proof gives us a doubly-exponential 
upper bound to the size of the (pseudo-)model for the input
MPNL formula (if any),
when interpreted in the linear orders under consideration.
Then, we devise an EXPSPACE decision procedure
for MPNL, interpreted over the integer numbers, and we show how
to adapt it to the cases of finite linear orders and natural numbers.
EXPSPACE-completeness immediately follows from the already known
EXPSPACE-hardness of the problem.
As a by-product, we solve the issue about the exact 
complexity of MPNL, with a binary encoding of 
metric constraints, interpreted over the natural numbers, 
which was left open in~\cite{sosym}. 
Moreover, since MPNL is expressively complete for a fragment of 
first-order logic with two variables and one successor function, 
interpreted over the same classes of linear orders~\cite{sosym}, 
the proposed decision procedure can be used to check the satisfiability 
of formulae of such a logic as well.

The paper is organized as follows. In Section \ref{sec:logic}, we 
introduce the logic. Then, in Section \ref{sec:types}, we provide 
some basic definitions and results to be used in the following. In 
Section \ref{sec:decfin}, we prove the decidability of the 
satisfiability problem for MPNL over finite linear orders.
In the following two sections, we generalize such a 
result to the cases of natural and integer numbers
by showing that every satisfiable formula has a model
that can be represented with a suitable small ``generator''.
Finally, in Section \ref{sec:decisionproc}, we outline an EXPSPACE
decision procedure for satisfiability checking in the most general
case of integer numbers, which can be easily tailored  to the
cases of finite linear orders and natural numbers.

%% file: Logic.tex
The logic MPNL can be viewed as a natural metric extension of PNL.
The language of PNL consists of a set $\AP$ of atomic propositions, the
propositional connectives $\neg$ and $\vee$, and the modal operators
$\dr$ and $\dl$ for Allen's relations \emph{meets} and
\emph{met by}, respectively~\cite{interval_relations}.
Representation theorems, axiomatic systems, and decidability
results for PNL, interpreted over various classes of linear orders,
have been given in~\cite{Bresolin09b,Goranko03a}. An optimal
tableau-based method for deciding the satisfiability problem
for the future fragment of PNL (RPNL) over the natural numbers has
been presented in \cite{bresolin07b}, and later extended to
the full PNL over the integers in~\cite{bresolin07a}, while
an optimal tableau system for RPNL over the class of all linear
orders can be found in~\cite{Bresolin08a}.
Optimal tableau-based decision procedures for PNL, interpreted over
various classes of linear orders, can be found in \cite{rr01}.

An extension of PNL, interpreted over the natural numbers, with
(a limited set of) metric constraints has been defined and
systematically studied in~\cite{sosym} (as a matter of fact,
a metric extension of RPNL was first considered in
\cite{rpnl_with_constraints}).
Let $\delta$ be the {\em distance} function over natural numbers
defined as $\delta(x,y) = |x-y|$ (the same definition applies to
any finite linear order and to the integer numbers).
Metric PNL (MPNL) is obtained
from PNL by adding a set of (pre-interpreted) atomic propositions
for length constraints. These propositions allow one to constrain
the length of the current interval and can be viewed as the
natural metric generalization of the modal constant $\pi$ of
propositional interval logics~\cite{Goranko03a}, which evaluates
to true precisely over point-intervals.
Formally, for each $\sim \in \{ <,\ \leq,\ =,\ \geq,\ > \}$, MPNL
features a length constraint $\len{\sim}{k}$, whose semantics is
defined as follows:
$M,[x,y] \mmodels \len{\sim}{k} \text{ iff } \delta(x,y) \sim k$.
Hereafter, we limit ourselves to one type of metric
constraints only, namely, $\lenlt{k}$, as all the
remaining ones can be expressed in terms of it.
As an example, we have that $M,[x,y] \mmodels \leneq{k}
\Leftrightarrow M,[x,y] \mmodels \lenlt{k+1} \wedge \neg \lenlt{k}$.
Formulae of MPNL (denoted by $\varphi,\psi, \ldots$) are generated by
the following grammar:
$$\varphi ::= \lenlt{k} \mid p \ |\ \neg\varphi\ |\ \varphi \vee
\varphi\ |\ \dl \varphi\mid \dr \varphi, \text{ where  $p\in \AP$
and $k\in \bbN$}.$$
The other propositional connectives, the logical constants
$\top$ ($true$) and $\bot$ ($false$), and the dual modal operators
$\Dr$ and $\Dl$ are defined as usual. Moreover, the modal constant
$\pi$ can be defined as $\lenlt{1}$.

\medskip

Given a linearly-ordered domain $\mathbb D  = \langle D,<\rangle $, a ({\em non-strict})
\emph{interval} over $\mathbb D$ is an ordered pair $[x,y]$, with $x
\leq y$. We denote by $\mathbb I(\mathbb D)$ the set of all intervals 
over $\mathbb D$. Moreover, we denote by $y_{max}$ the greatest point in $D$ (if there 
is not such a point, we put $y_{max}=+\infty$) and by $y_{min}$ the 
least point in $D$ (if there is not such a point, we put $y_{min}=-\infty$).
The semantics of MPNL is given in terms of {\em models} of the form
$M=\langle\mathbb D,V\rangle$, where $V: \AP
\rightarrow 2^{\mathbb I(\mathbb D)}$ is a valuation function
assigning a set of intervals to every atomic proposition. From
now on, we assume the domain $D$ to be either $\bbZ$, $\bbN$, or
a finite prefix of $\bbN$.
We recursively define the truth relation $\mmodels$ as follows:

\begin{itemize}

\item $M,[x,y]\mmodels p$ iff $[x,y] \in V(p)$, for any $p\in\AP$;
\item $M,[x,y]\mmodels len_{<k}$ iff $\delta(x,y)<k$;
\item $M,[x,y]\mmodels \neg\varphi$ iff it is not the case that $M,[x,y]\mmodels\varphi$;
\item $M,[x,y]\mmodels \varphi\vee\psi$ iff $M,[x,y]\mmodels\varphi$ or $M,[x,y]\mmodels\psi$;
\item $M,[x,y]\mmodels \dl\varphi$ iff there exists $z\le x$ such that $M,[z,x]\mmodels\varphi$;
\item $M,[x,y]\mmodels \dr\varphi$ iff there exists $z\ge y$ such that $M,[y,z]\mmodels\varphi$.
\end{itemize}

\noindent An MPNL-formula $\varphi$ is said to be {\em satisfiable} if there exist a 
model $M=\langle\mathbb D,V\rangle$ and an interval $[x,y] \in \mathbb I(\mathbb D)$ 
such that $M,[x,y]\mmodels\varphi$.

\medskip

In~\cite{sosym}, the satisfiability problem for MPNL has been shown to be decidable when
interpreted over the set of natural numbers. More precisely, it has been shown that the
satisfiability problem for MPNL over the set of natural numbers is NEXPTIME-complete when
either the maximal $k$ that occurs in metric constraints is a constant or the parameter $k$
of metric constraints is represented in unary, and it is in between EXPSPACE and 2NEXPTIME when
the parameter $k$ is represented in binary. In the following, we will show that the 
satisfiability problem for MPNL, with a binary encoding of metric constraints, 
interpreted over finite linear orders, the natural 
numbers, and the integer numbers, is actually EXPSPACE-complete, by developing an 
EXPSPACE decision procedure for it. It is worth noticing that the model-theoretic 
argument behaves, in a way, worse than the one in~\cite{sosym}, as it 
provides a doubly-exponential upper bound on the size of (pseudo-)models,
regardless of the representation of $k$. Nevertheless, we will show that
in the search for a (pseudo-)model of a given formula, at any time it suffices 
to keep track of a portion of it that can be recorded in exponential 
space, thus obtaining an EXPSPACE decision procedure. 

%% file: Compass.tex
In this section, we introduce the basic logical machinery to be used in the following
sections.
Let $M=\langle\mathbb D, V\rangle$ be a model for an MPNL-formula $\varphi$.
In the sequel, we relate every interval in $M$ to the set of sub-formulae of $\varphi$
it satisfies. To do that, we introduce the key notions of $\varphi$-{\em atom} and
$\varphi$-{\em type}.
First of all, we define the \emph{closure} $\closure(\varphi)$ of $\varphi$ as the
set of all sub-formulae of $\varphi$ and of their negations (we identify $\neg\neg\alpha$
with $\alpha$, $\neg\dr\alpha$ with $\Dr\neg\alpha$, and so on), and we define
$\cK_{\varphi}=\{k\ |\ len_{<k}\in \closure(\varphi)\}$ as the set of all metric
parameters that appear in $\varphi$.

\begin{definition}\label{def:atom}
A \emph{$\varphi$-atom} is any non-empty set $F\subseteq \closure(\varphi)$ such that:
\begin{compactenum}
\item for every $\alpha\in\closure(\varphi)$, we have $\alpha\in F$ iff $\neg\alpha\nin F$,
\item for every $\gamma=\alpha\vel\beta\in\closure(\varphi)$, we have $\gamma\in F$ iff $\alpha\in F$ or $\beta\in F$, and
\item for every $k,k'$ in $\cK_{\varphi}$ such that $k < k'$, we have that $len_{<k}\in A$ implies $len_{<k'}\in A$.
\end{compactenum}
\end{definition}

Intuitively, a \emph{$\varphi$-atom} is a maximal {\sl locally consistent} set of formulas chosen
from $\closure(\varphi)$. Note that the cardinality of  $\closure(\varphi)$ is {\sl linear} in
the length $\lenf{\varphi}$ of $\varphi$, while the number of $\varphi$-atoms is {\sl at most
exponential} in $\lenf{\varphi}$ (precisely,
we have that $\lenf{\closure(\varphi)}$ is at most $2\lenf{\varphi}$
and there are at most $2^{\lenf{\varphi}}$ distinct atoms).
We define $\cA_\varphi$ as the set of all possible atoms that can
be built over $\closure(\varphi)$.
For every model $M$ and every interval $[x,y]\in\mathbb I(\mathbb D)$,
we associate the set of all formulas $\psi\in\closure(\varphi)$
such that $M,[x,y]\sat\psi$ with $[x,y]$. We call such a set the
\emph{$\varphi$-type} of $[x,y]$ and we denote it by $\type_M([x,y])$.
We have that every $\varphi$-type is a $\varphi$-atom, but not vice
versa. Hereafter, $\varphi$-atoms (resp., $\varphi$-types) will be
simply called atoms (resp., types). Given an atom $F$, we denote by
$\obs_r(F)$ (resp., $\obs_{l}(F)$ ) the set of all \emph{future}
(resp., \emph{past}) \emph{observable formulae} of $F$, namely, the
set of formulae $\psi\in F$ such that $\dr\psi\in \closure(\varphi)$
(resp., $\dl\psi\in \closure(\varphi)$). Similarly, given an atom $F$, we denote by
$\req_r(F)$ (resp., $\req_{l}(F)$) the set of all \emph{$\dr$-requests}
(resp., \emph{$\dl$-requests}) of $F$, namely, the set of formulae
$\psi\in\closure(\varphi)$ such that $\dr \psi\in F$ (resp., $\dl
\psi\in F$), and we use the shorthand $\req(F)$ for
$\req_r(F)\cup\req_{l}(F)$.
Making use of the above notions, we can define the following
relation between two atoms $F$ and $G$:
$$
\begin{array}{rcl}
  F \dep{R} G
  &\;\text{iff}\;&
  \obs_{r}(G) \,\subseteq\, \req_r(F) \mbox{ and } \obs_{l}(F) \,\subseteq\, \req_{l}(G)
\s1
\end{array}
$$
The relation $\dep{R}{}$ satisfies a \emph{view-to-type dependency},
that is, for every pair of intervals $[x,y],[x',y']$ in
$\mathbb{I(D)}$, we have that  $ y = x'
\;\text{implies }   \type_M([x,y]) \,\dep{R}{}\, \type_M([x',y'])$.

\begin{figure}[!!t]
\centering
\begin{tikzpicture}[scale=0.8]
\draw[very thick,|-|] (0,0) -- (2,0)node[pos=0.5, above] {$[x_0,y_0]$};
\draw[very thick,|-|,red] (-3,0.8) -- (2,0.8)node[pos=0.5, above=0.001cm,black] {$[x_3,y_3], \neg len_{<k+1}$};
\draw[very thick,|-|,blue] (2,-0.8) -- (4,-0.8)node[pos=0.5, above=0.001cm,black] {$[x_1,y_1]$};
\draw[very thick,|-|,green] (-3,-0.8) -- (0,-0.8)node[pos=0.5, above=0.001cm,black] {$[x_2,y_2], len_{<k}$};
\pgftransformshift{\pgfpoint{8cm}{-0.5cm}}
\draw[step=0.5cm,gray,very thin] (-2.9,-2.9) grid (2.9,2.9);
\fill[color=white] (-2.9,-2.9) -- (2.9,2.9) -- (2.9,-2.9);
\draw (-2.9,-2.9) -- (2.9,2.9);
\draw[dashed] (-0.5,1) -- (-0.5, -0.5);
\draw[dashed] (1,2) -- (1, 3);
\draw[dashed] (-0.5,1) -- (1, 1);
\draw[dashed] (1,2) -- (1, 1);
\draw[dashed] (-3,-0.5) -- (-0.5, -0.5);
\draw[dashed] (-2,1) -- (-2, -2);

\node[shape=circle,draw=black,inner sep=2pt,fill=black, label={above :$(x_0,y_0)$}](A) at (-0.5,1) {};
\node[shape=circle,draw=black,inner sep=2pt,fill=red, label={left :$(x_3,y_3)$}](A) at (-2,1) {};
\node[shape=circle,draw=black,inner sep=2pt,fill=blue, label={above :$(x_1,y_1)$}](A) at (1,2) {};
\node[shape=circle,draw=black,inner sep=2pt,fill=green, label={above right:$(x_2,y_2)$}](A) at (-2,-0.5) {};

\draw[decorate,decoration={brace,amplitude=0.4cm}, color=black] 	(-2,-2) -- (-2,0.5) node[pos=0.5,left=0.3cm]{$k$};

\end{tikzpicture}
\caption{Correspondence between intervals and the points of the compass structure.}
\label{fig:compassstructure}
\end{figure}
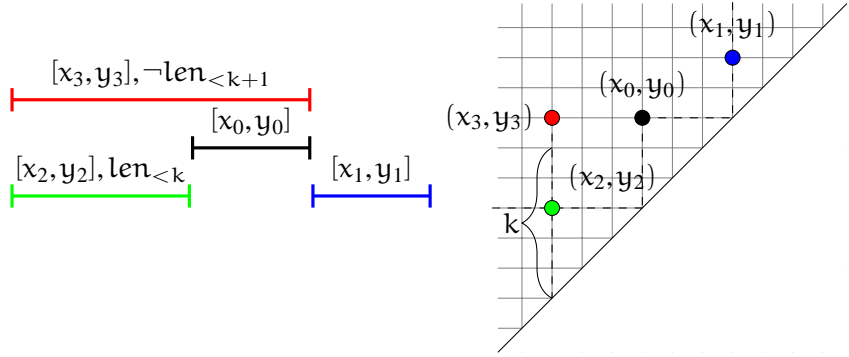

\medskip

We provide now a natural interpretation of MPNL over grid-like
structures ({\em compass structures}) by exploiting the existence
of a natural bijection between the intervals $[x,y]$ and the
points $(x,y)$ of a $D\times D$ grid with $x \leq y$. Such
an interpretation was originally proposed by Venema in~\cite{compass_logic},
and it can be given for HS and all its fragments as well.
As an example, Figure~\ref{fig:compassstructure} shows four intervals
$[x_0,y_0],...,[x_3,y_3]$ such that (i) $y_0 = x_1$, (ii)
$x_0 = y_2$, (iii) the length of $[x_2,y_2]$ is less
than $k$, and (iv) the length of $[x_3,y_3]$ is greater than $k$, together
with the corresponding points $(x_0,y_0),...,(x_3,y_3)$ of the grid (notice that
Allen's interval relations \emph{meets} and \emph{met by} are mapped into the
corresponding spatial relations between pairs of points).
Such an alternative interpretation of MPNL over compass
structures will be exploited in the decidability proofs
to make them easier to understand.

\begin{definition}\label{def:compassstructure}
Given an $MPNL$ formula $\varphi$, a  \emph{compass}
$\varphi$-\emph{structure} is a pair $\cG=(\bbP_\bbD,\cL)$,
where $\bbP_\bbD$ is the set of points of the form $(x,y)$,
with $x,y \in D$ and $ x\leq y$, and $\cL$ is a function that
maps any point $(x,y)\in\bbP_\bbD$ to a $\varphi$-atom $\cL(x,y)$
in such a way that:
\begin{dotlist}
  \item for every pair of points $(x,y),(x',y')\in\bbP_\bbD$ ,
        if $y = x'$ then $\cL(x,y) \dep{R} \cL(x',y')$
        ({\bf temporal consistency});
  \item for every point $(x,y) \in\bbP_\bbD$,
  			and every $len_{<k}\in \cL(x,y)$, $y-x<k$ ({\bf
  			length consistency}).
\end{dotlist}
\end{definition}

We say that a compass $\varphi$-structure $\cG=(\bbP_\bbD,\cL)$
\emph{features} a formula $\psi$ if there exists a point $(x,y)\in\bbP_\bbD$
such that $\psi \in \cL(x,y)$. Fulfilling compass structures are defined as follows.

\begin{definition}\label{def:fulfillingcompass}
Given an $MPNL$ formula $\varphi$ and  compass $\varphi$-structure
$\cG=(\bbP_\bbD,\cL)$ for it, we say that $\cG$ is \textbf{fulfilling}
if and only if for every point $(x,y)\in\bbP_\bbD$ and every formula 
$\psi\in\req_r\bigl(\cL(x,y)\bigr)$ (resp., $\psi\in\req_l\bigl(\cL(x,y)\bigr)$), 
there exists a point $(x',y')\in\bbP_\bbD$ such that $x'=y$ (resp., $y'=x$) 
and $\psi \in \cL(x',y')$.
\end{definition}

The following proposition proves that the satisfiability problem for
$MPNL$ is reducible to the problem of deciding, for any given formula
$\varphi$, whether there exists a compass $\varphi$-structure featuring
$\varphi$. Its easy proof is left to the reader.

\begin{proposition}\label{prop:compassstructure}
An $MPNL$-formula $\varphi$ is satisfiable if
and only if there exists a fulfilling compass $\varphi$-structure that
features $\varphi$.
\end{proposition}
Without loss of generality, we will assume $\varphi$ to be satisfied
by the initial point-interval $0$ (resp., to belong to $\cL(0,0)$)
\cite{thesispietro}.

Given an MPNL-formula $\varphi$, we denote by $\kphi$ the maximum
$k$ occurring in $\varphi$. If there is not any $k$ in $\varphi$,
we simply put $k_{\varphi}=0$.
We assume $\kphi$, as well as any length constraint occurring in
$\varphi$, to be encoded in binary, and thus it immediately follows that
$\kphi\leq 2^{|\varphi|}$.

Given a compass $\varphi$-structure $\cG=(\bbP_\bbD,\cL)$, we define
a \emph{marking  function} $\cM: \bbP_\bbD\rightarrow \cA_\varphi \times
2^{\closure(\varphi)} \times \{0,\ldots,\kphi\}$ such that, for every
$(x,y)\in \bbP_\bbD$, $\cM(x,y)= (F, \Psi, h)$, where (i) $F = \cL(x,y)$,
(ii) $\Psi = \{ \psi\in \closure(\varphi) \ |\  \psi \in \req_r(x,x) \wedge
\forall x \leq y' \leq y  (\psi \notin \cL(x,y')) \}$, and (iii) $h$ is
defined as follows:
\[ h =\begin{cases} y-x & \text{ if $y-x<\kphi$; } \\ \kphi
& \text{otherwise.} \end{cases}\]
Notice that, for every point $(x,y)$, $\Psi$ is the set of formulae that must
belong to the labeling of points $(x,y')$, with $y' > y$ (points ``above'' 
$(x,y)$), to guarantee the fulfilling of all $\dr$-requests in $\cL(x,x)$,
that is, for each $\psi \in \Psi$, there must exist at least one point 
$(x,y')$ such that $\psi \in \cL(x,y')$).

Let $\cA_\varphi^\cM$ be the image of $\cM$. We call any triplet in
$\cA_\varphi^\cM$ a \emph{marked atom}. It can be easily shown that
$|\cA_\varphi^\cM| \leq 2^{3|\varphi|}$ ($|A_\varphi| \leq 2^{|\varphi|}$,
$|\req_r(\cL(x,x))| \leq |\varphi|$, and $\kphi\leq 2^{|\varphi|}$).


\begin{definition}\label{def:horconfiguration}
Given an $MPNL$ formula $\varphi$, a compass $\varphi$-structure
$\cG=(\bbP_\bbD,\cL)$ for $\varphi$, and $y\in D$, we define
the \emph{horizontal configuration} of $y$ in $\cG$ as a counting
function $\cC_y: \cA^\cM_\varphi\rightarrow \bbN \cup\{\omega\}$
such that for every $(F, \Psi, h) \in \cA^\cM_\varphi$,
$\cC_y(F, \Psi, h) = |\{ x \ | \ \cM(x,y)= (F, \Psi, h)\}|$.
\end{definition}

\noindent It is worth noticing that, for any given $y$, (i) there exists
a unique marked atom of the form $(F, \Psi, 0)$, with $\cC_y(F, \Psi, 0)=1$, and
(ii) for every $0 < h<\kphi$, there exists at most $1$ marked
atom of the form $(F,\Psi,h)$, and if for every marked atom $(F,\Psi,h)$,
$\cC(F,\Psi,h)=0$, then $\cC(F',\Psi',h') =0$ for every marked
atom $(F',\Psi',h')$ with $h' > h$.  On the contrary, there is not
a bound on the number of occurrences of a marked node of the
form $(F, \Psi, \kphi)$ (it can be equal to $\omega$).

Finally, we define the following equivalence relation on the
set of horizontal configurations, where $p$ and $f$ are defined
as $p=|\{ \dl \psi \in \closure(\varphi)  \}|$ and $f=|\{ \dr
\psi \in \closure(\varphi) \}|$, respectively.

\begin{definition}\label{def:horconfequiv}
Given an $MPNL$ formula $\varphi$ and a compass $\varphi$-structure
$\cG=(\bbP_\bbD,\cL)$ for it, we say that two horizontal configurations
$\cC_y$  and $\cC_{y'}$  are \emph{equivalent} (written $\cC_y\equiv\cC_{y'}$)
if and only if for every $(F,\Psi,h)\in \cA^\cM_\varphi$, either
$\cC_{y'}(A,\Psi,h)=\cC_{y}(F,\Psi,h)$ or ($h = \kphi$ and) both $\cC_{y}(F,\Psi,\kphi)
\geq p \cdot f + p$  and $\cC_{y'}(F,\Psi,\kphi) \geq p \cdot f + p$.
\end{definition}

It can be easily shown that $\equiv$
is an equivalence relation of finite index. For every marked atom 
$(F,\Psi, h)\in\cA^\cM_\varphi$, we do not distinguish between two 
configurations $\cC_y$ and $\cC_{y'}$ such that $\cC_y(F,\Psi, h)$ 
and $\cC_{y'}(F,\Psi, h)$ are different, but both greater than or equal to
$p \cdot f+p$. Hence, the number of equivalence classes in $\equiv$ 
is bounded by
\begin{equation*}
\Bigl(p\cdot f+p+1\Bigr)^{\left|\cA^\cM_\varphi\right|} \leq 
\left( \frac{|\varphi|^2}{4} + \frac{|\varphi|}{2} +1 \right)^{2^{3|\varphi|}},
\end{equation*}
since $p\cdot f+p\leq  \frac{|\varphi|^2}{4} + \frac{|\varphi|}{2}$ and $\left|\cA^\cM_\varphi\right| \leq 2^{3|\varphi|}$.

%% file: Decidabilityfinite.tex
\section{Decidability of MPNL over finite linear orders}\label{sec:decfin}

\begin{figure}
\begin{tikzpicture}
\end{tikzpicture}
\end{figure}

In this section, we show that if there exists a finite fulfilling 
compass structure $\cG$ for an MPNL formula $\varphi$, then there 
exists a finite fulfilling compass structure $\cG'$ whose size is 
at most doubly exponential in the length of $\varphi$.
To prove this result, we will make use of the following lemma, 
which states that we can always shrink the size of a fulfilling 
compass structure, provided that there exist $y, y'$ such that $\cC_y
\equiv\cC_{y'}$.

\begin{lemma}\label{lem:removal}
Let $\varphi$ be an $MPNL$ formula and let $\cG=(\bbP_\bbD,\cL)$ be a 
finite fulfilling compass $\varphi$-structure which features $\varphi$. 
If there exist $\oy, \oy' \in D$, with $\oy<\oy'$, such that $\cC_{\oy}
\equiv \cC_{\oy'}$, then it is possible to build a finite fulfilling compass 
$\varphi$-structure $\cG'=(\bbP_{\bbD'},\cL')$ featuring $\varphi$ with $|D'|=
|D|-(\oy'-\oy)$.
\end{lemma}

\begin{proof}
Suppose that $\cG=(\bbP_\bbD,\cL)$ is a finite fulfilling compass 
$\varphi$-structure which features $\varphi$ and such that there exist $\oy, \oy' \in D$, 
with $\oy<\oy'$, such that $\cC_{\oy} \equiv \cC_{\oy'}$. 
We build a compass $\varphi$-structure $\cG'=(\bbP_{\bbD'},\cL')$, with $|D'|
=|D|-(\oy'-\oy)$, by executing the following procedure.

\begin{enumerate}
\item For every $(x,y)\in \bbP_{\bbD'}$, with $y\leq\oy$, we put $\cL'(x,y)=
\cL(x,y)$.

\item For every $(x,y)\in \bbP_{\bbD'}$, with $y>\oy$ and $\oy-\kphi< x\leq y$, 
we put $\cL'(x,y)=\cL(x+(\oy'-\oy), y+(\oy'-\oy))$.

\item For every $(A,\Psi,\kphi)\in \cA^{\cM}_{\varphi}$, we define a partial 
injective function $g: \{0,\ldots,\oy-\kphi \}\rightarrow\{0,\ldots,\oy'-\kphi \}$ 
as follows:
\[ g(x)=
\begin{cases} 
x' \text{ with  $\cM(x',\oy')=(A,\Psi,\kphi )$ } & \text{ if  
$\cM(x,\oy)=(A,\Psi,\kphi )$ and } \\ 
& \text{ $\cC_{\oy}(A,\Psi,\kphi)=\cC_{\oy'}(A,\Psi, \kphi)$ } \\ 
undefined & \text{ otherwise  } 
\end{cases}  
\]
By injectivity of $g$, every $x$ (where $g$ is defined) is associated with 
a distinct $x'$. Moreover, since $\cC_{\oy}(A,\Psi,\kphi)=\cC_{\oy'}(A,\Psi, \kphi)$, 
for every $x'$ such that $\cM(x',\oy')=(A,\Psi,\kphi)$, there exists (a unique) $x$ 
such that $g(x) = x'$. Now, for every $0\leq x\leq \oy-\kphi$ such that $g(x)$ is 
defined, we put $\cL'(x,\oy+i)=\cL(g(x),\oy'+i)$ for every $1\leq i\leq y_{max}-\oy'$.


\item For every $(A,\Psi,\kphi)\in \cA^{\cM}_{\varphi}$ such that $\cC_{\oy'}(A,\Psi,
\kphi)\geq p\cdot f+p$, we choose a ``witness'' $w_{(A,\Psi)}$ such that $\cM(w_{(A,
\Psi)},\oy')=(A,\Psi,\kphi)$. Then, we identify a minimal set of \emph{essential 
elements} $\Essentials^{\oy'}_{(A,\Psi)}= \{ y'_1,\ldots,y'_{m} \}$ such that, 
for every $\psi\in \Psi$, there exists a point $y'_{j}\in \Essentials^{\oy'}_{(A,\Psi)}$ 
with $\psi \in \cL(w_{(A,\Psi)},y'_j)$. As $|\Psi| \leq f$, it immediately follows
that $m\leq f$. Moreover, by definition of (the second component of a) marked atom,
$y'_i > \oy'$ for every $1 \leq i \leq m$. 
%
%
%
%
Now, let $\Blocked^{\oy'}_{(A,\Psi)} = \{ x'_1,\ldots,x'_{m'} \}$ be a minimal set 
of elements, called \emph{blocked elements}, satisfying the following condition:
for every  $1\leq i \leq m$ and every $\psi \in \req_l(y'_i,y'_i)$, if there exists 
$x' \in D$ such that $\psi \in \cL(x',y'_i)$ and $\cM(x', \oy')=(A,\Psi,\kphi)$,
then there exists $x_j' \in \Blocked^{\oy'}_{(A,\Psi)}$ such that $\psi \in \cL(x'_j,
y'_i)$ and $\cM(x'_j, \oy')=(A,\Psi,\kphi)$. As $m \leq f$ and $|\req_l(y'_i,y'_i)| 
\leq p$, $|\Blocked^{\oy'}_{(A,\Psi)}|\leq p \cdot f$.
Since $\cC_{\oy} \equiv \cC_{\oy'}$, a set $\Blocked^{\oy}_{(A,\Psi)}=
\{ x_{1},\ldots, x_{m'} \}$ exists such that, for every $1\leq i\leq m'$, 
$\cM(x_i,\oy) = (A,\Psi,\kphi) (= \cM(w_{(A,\Psi)},\oy'))$. 
For every $1\leq i\leq m'$ and every $1\leq j\leq y_{max}-\oy'$, we put $\cL'(x_i,
\oy+j)=\cL(x'_i,\oy'+j)$ . 
In such a way, all points $(x_i,y)$ in $\cG'$, with $1\leq i \leq m'$, turn out to be
labeled and all $\dr$-requests of points $(x_i,x_i)$ are fulfilled.
\item Once the above steps have been executed, there may exist some 
$x \in D$ such that the labeling of points $(x,y) \in \bbP_{\bbD'}$,
with $y>\oy$, is still undefined. Let $\cM(x, \oy) = (A,\Psi,\kphi)$.
By construction, $\cC_{\oy}(\cM(x, \oy))\geq p \cdot f+p$. For every
unlabeled point $(x,y)$, we put $\cL'(x, y) = \cL(w_{(A, \Psi)}, y + 
(\oy'-\oy))$, where $w_{(A,\Psi)}$ is the witness chosen at step 4.
\end{enumerate}

Unfortunately, there is no guarantee that all $\dl$-requests 
are fulfilled in $\cG'$. Let $y>\oy$ such that there exists 
$\psi \in \req_{l}(\cL'(y,y))$ which is not fulfilled in $\cG'$. 
By construction, $\cL'(y,y) = \cL(y+(\oy'-\oy), y+(\oy'-\oy))$, 
and thus, since $\cG$ is fulfilling, there exists a point 
$(x'_\psi, y+(\oy'-\oy))$ such that $\psi \in \cL(x'_\psi,
y+(\oy'-\oy))$. We must distinguish two cases:
\begin{compactenum}
\item[a)] for every witness $w_{(A,\Psi)}$, $y+(\oy'-\oy)\notin 
\Essentials^{\oy'}_{(A,\Psi)}$. Let $(A,\Psi, \kphi)$ be the 
marked atom associated with $(x'_\psi, \oy')$ in $\cG$, that is,
$\cM(x'_\psi, \oy') = (A,\Psi, \kphi)$. It holds that 
$\cC_{\oy'}(\cM(x'_\psi, \oy'))\geq p\cdot f+p$ (if this was not 
the case,  $x'_\psi$ would not belong to the range of $g$, thus 
violating the properties we impose on it at step 3), and thus
$\cC_{\oy}(\cM(x'_\psi, \oy))\geq p\cdot f+p$ as well. Since
$|\Blocked^{\oy}_{(A,\Psi)}| (= |\Blocked^{\oy'}_{(A,\Psi)}|) 
\leq p \cdot f$, there exist at least $p$ elements $x_{m'+1}, 
\ldots, x_{m'+p}$ such that, for $1 \leq i \leq p$, $x_{m'+i} 
\not\in \Blocked^{\oy}_{(A, \Psi)}$ and $\cM(x_{m'+i}, \oy) 
(= \cM(w_{(A,\Psi)}, \oy')) = (A,\Psi,\kphi)$. We show that,
in order to fulfill $\psi$, the labeling 
of at least one among $(x_{m'+1},y), \ldots, (x_{m'+p},y)$ 
can be suitably updated. To this end, it suffices to observe 
that $|\req_l(\cL'(y,y))|\leq p$ and thus there exists $1 \leq 
j \leq p$ such that, for every $\theta \in \req_l(\cL'(y,y))$,
if $\theta \in \cL'(x_{m'+j},y)$, then $\theta \in \cL'(x_{m'+l},
y)$, for some $0 \leq l \leq p$, with $l \neq j$,  as well.
Moreover, since $y+(\oy'-\oy)\notin \Essentials^{\oy'}_{(A,\Psi)}$,
for every $\phi \in \Psi$, there exists $y' (> \oy) \neq y$ such 
that $\phi \in \cL'(x_{m'+j},y')$ and thus $(x_{m'+j},y)$ is not needed to 
fulfill $\dr$-requests in $\req_r(\cL'(x_{m'+j},x_{m'+j}))$. Hence, we can safely
revise $\cL'(x_{m'+j},y)$ putting $\cL'(x_{m'+j},y)=\cL(x'_\psi,y+(\oy'-\oy))$;

\item[b)] there exists a witness $w_{(\overline{A},\overline{\Psi})}$ 
such that $y+(\oy'-\oy)\in \Essentials^{\oy'}_{(\overline{A},
\overline{\Psi})}$. Let $(A,\Psi, \kphi)$ be the marked atom 
associated with $(x'_\psi, \oy')$ in $\cG$, and let $(x_{m'+1},y), 
\ldots, (x_{m'+p}, y)$ be the $p$ elements of case a). As above, 
we can show that, to fulfill $\psi$, the labeling of at least 
one among them, say $(x_{m'+j},y)$, can be suitably updated. 
The irrelevance of $(x_{m'+j},y)$ with 
respect to requests in $\req_l(\cL'(y,y))$ can be proved in exactly 
the same way. To complete the proof, it suffices to show that 
$y+(\oy'-\oy) \notin \Essentials^{\oy'}_{(A,\Psi)}$. By 
contradiction, assume that $y+(\oy'-\oy)\in \Essentials^{\oy'}_{(A,\Psi)}$.
This implies that there exists $x_{i}\in \Blocked^{\oy}_{(A,
\Psi)}$ such that $\psi \in \cL'(x_{i}, y)$, and thus 
$\psi$ is fulfilled in $\cG'$ (contradiction). Then, we 
can proceed as in case a) and rewrite $\cL'(x_{m'+j},y)$ 
as $\cL(x'_\psi,y+(\oy'-\oy))$.
\end{compactenum}
$\cG'$ is a fulfilling compass $\varphi$-structure for $\varphi$.
\end{proof}

By exploiting Lemma \ref{lem:removal}, we can prove that a formula $\varphi$ is satisfiable 
by a finite compass structure iff it is satisfiable by a finite compass 
structure whose horizontal configurations are pairwise non-equivalent. 

\begin{theorem}\label{teo:finiterepresentation}
Let $\varphi$ be an $MPNL$-formula. If there exists a finite fulfilling  
compass $\varphi$-structure $\cG=(\bbP_\bbD,\cL)$  which features $\varphi$, 
then there exists a finite fulfilling compass $\varphi$-structure $\cG'=
(\bbP_{\bbD'},\cL')$ featuring $\varphi$ such that $\lenf{D'} \leq 
\left( \frac{|\varphi|^2}{4} + \frac{|\varphi|}{2} +1 \right)^{2^{3|\varphi|}}$.
\end{theorem}

\begin{proof}
Let $\cG=(\bbP_\bbD,\cL)$ be a finite fulfilling compass $\varphi$-structure featuring 
$\varphi$ and suppose that $\lenf{D} > \left( \frac{|\varphi|^2}{4} + 
\frac{|\varphi|}{2} +1 \right)^{2^{3|\varphi|}}$. Since the index of $\equiv$ 
is smaller than $\lenf{D}$, there exist $\oy, \oy' \in D$, with
$\oy<\oy'$, such that $\cC_y\equiv \cC_y'$. Then, we exploit Lemma~\ref{lem:removal} 
to build a smaller compass $\varphi$-structure $\cG_1 = (\bbP_{\bbD_1},\cL_1)$ with 
$\lenf{D_1} = \lenf{D} - (\oy' - \oy)$. By iterating such a contraction step, we 
eventually obtain a compass $\varphi$-structure $\cG_n = (\bbP_{\bbD_n},\cL_n)$ whose horizontal 
configurations are pairwise non-equivalent. Since the number of equivalence classes 
in $\equiv$ is less than or equal to $\left( \frac{|\varphi|^2}{4} + \frac{|\varphi|}{2} 
+1 \right)^{2^{3|\varphi|}}$, the thesis immediately follows.
\end{proof}

%% file: Decidabilitynatural.tex
\section{Decidability of MPNL over the naturals}\label{sec:decnat}

We now extend the result of the previous section to cope with 
the satisfiability problem for MPNL over $\bbN$.
First, we identify a subset of finite compass $\varphi$-structures, called 
\emph{compass generators}, which turn out to be crucial for 
decidability.

\begin{definition}\label{def:compassgennat}
Let $\varphi$ be an MPNL formula. An \emph{$\bbN$-compass generator} 
for $\varphi$ is a finite compass $\varphi$-structure 
$\cG=(\bbP_\bbD,\cL)$, which features $\varphi$, that 
satisfies the following conditions:
\begin{compactenum}
	\item all $\dl$-requests of every point $(x,y)\in\bbP_\bbD$ are fulfilled;
	\item there exists $y_{inf}$, with $y_{max} - y_{inf} \geq k_\varphi$, such that:
		\begin{compactenum}
			\item for every $(F,\Psi, h) \in \cA^{\cM}_\varphi$, if $\cC_{y_{max}}(F,\Psi, h) >0$,
			then $\cC_{y_{inf}}(F,\Psi, h) >0$, and
			\item $\cM(x,y_{max})=(F,\emptyset, h)$, for every $0\leq x \leq y_{inf}$.
		\end{compactenum}
\end{compactenum}
\end{definition}

\begin{theorem}\label{teo:satiffcompassgen}
An MPNL formula $\varphi$ is satisfiable over $\bbN$ iff there exists
an $\bbN$-compass generator for it.
\end{theorem}

\begin{proof}
To prove the left-to-right direction, suppose $\varphi$ to be satisfiable over 
$\bbN$, and let $\cG=(\bbP_{\bbN}, \cL)$ be a fulfilling compass $\varphi$-structure  
which features $\varphi$. 
Since the index of $\equiv$ is finite, there must exist an infinite sequence $\cS=y_1<y_2<\ldots$ in $\bbN$ 
such that $\cC_{y_i}\equiv\cC_{y_j}$ for every $i,j\in\bbN$. Consider now the 
first element $y_1$ in $\cS$, and let $(x,y_1)\in\bbP_{\bbN}$ be 
a point on the row $y_1$. Suppose $\cM(x, y_1) = (F,\Psi,\kphi)$. Since $\cG$ 
is fulfilling, for every $\psi \in \Psi$, there exists $y_\psi > y_1$ 
such that $\psi \in \cL(x,y_\psi)$. Let $\oy$ be the maximum of such $y_\psi$ 
with respect to every $x \leq y_1$ and every $\psi\in \Psi$, and 
let $y_j$ be the smallest element in $\cS$ such that $\oy < y_j$ and $y_j - y_1 \geq 
\kphi$. 
By the definition of the marking function $\cM$, we have that $\cM(x,y_j)=(F,\emptyset, 
h)$, for every $0\leq x \leq y_1$. Consider now the restriction $\cG'$ of $\cG$ to 
$D = \{0,1,\ldots,y_j\}$. It is straightforward to check that, given $y_{max} = y_j$,
$y_1$ satisfies the conditions for $y_{inf}$ of Definition~\ref{def:compassgennat}, 
and thus $\cG'$ is an $\bbN$-compass generator featuring $\varphi$ ($(0,0)$ belongs 
to $\cG'$).

To prove the right-to-left direction, suppose that $\cG=(\bbP_{\bbD}, \cL)$ is 
an $\bbN$-compass generator for $\varphi$. We build a fulfilling compass $\varphi$-structure 
$\cG_\omega=(\bbP_{\bbN}, \cL_\omega)$ as the (infinite) union of an appropriate 
sequence of $\bbN$-compass generators $\cG_0 \subset \cG_1 \subset \ldots$. 
First, we take $\cG_0 = \cG$. Then, for every $i \geq 0$, we build $\cG_{i+1}=
(\bbP_{\bbD_{i+1}}, \cL_{i+1})$ starting from $\cG_i=(\bbP_{\bbD_i}, \cL_i)$ as 
follows. Let $y_{inf} \in D_{i}$ satisfy the conditions of 
Definition~\ref{def:compassgennat}. We put $D_{i+1} = \{0, 1, \ldots, y_{max}, 
\dots, y_{max} + (y_{max} - y_{inf})\}$ and we define $\cL_{i+1}$ as follows:
\begin{compactenum}
	\item for every $(x,y) \in \bbP_{\bbD_i}$, we put $\cL_{i+1}(x,y) = \cL_i(x,y)$;
	
	\item for every $(x,y) \in \bbP_{\bbD_{i+1}}$ such that $x > y_{max} - \kphi$
	and $y > y_{max}$, we put $\cL_{i+1}(x,y) = \cL_i(x-(y_{max} - y_{inf}),
	y-(y_{max} - y_{inf}))$;
	
	\item for every $(x,y) \in \bbP_{\bbD_{i+1}}$ such that $y_{inf} - \kphi \geq x \geq 0$
	and $y > y_{max}$, we put $\cL_{i+1}(x,y) = \cL_i(x,y-(y_{max} - y_{inf}))$;

	\item for every $(x,y) \in \bbP_{\bbD_{i+1}}$ such that $y_{max} - \kphi \geq x > 
	y_{inf} - \kphi$ and $y > y_{max}$, we put $\cL_{i+1}(x,y) = \cL_i(x',y-(y_{max} - y_{inf}))$,
	for some $x'$ such that $\cM(x',y_{inf}) = \cM(x,y_{max})$ (the existence of
	such an $x'$ is guaranteed by property (a) of Definition \ref{def:compassgennat}).
\end{compactenum}
By construction, for every $(F,\Psi, h) \in \cA^{\cM}_\varphi$,
if $\cC_{y_{max}+(y_{max} - y_{inf})}(F,\Psi, h) >0$, then $\cC_{y_{max}}(F,\Psi, h) >0$, 
Moreover, $\cM(x,y_{max}+(y_{max} - y_{inf}))=(A,\emptyset, h)$, for every $0\leq x \leq y_{max}$, 
and thus $\cG_{i+1}$ is a $\bbN$-compass generator for $\varphi$. 

The fulfilling compass $\varphi$-structure satisfying $\varphi$ on $\bbN$ we were looking for
is $\cG_\omega= \bigcup_{i\geq0} \cG_i$.
\end{proof}

\begin{theorem}\label{teo:boundcompassgen}
Let $\varphi$ be an MPNL formula. If there exists an $\bbN$-compass generator 
$\cG=(\bbP_\bbD,\cL)$ that features $\varphi$, then there exists an $\bbN$-compass 
generator $\cG'=(\bbP_{\bbD'},\cL')$, that features $\varphi$, with $|D'|\leq 
\left(2^{3|\varphi|}+2\right) \cdot \left( \frac{|\varphi|^2}{4} + \frac{|\varphi|}{2} 
+1 \right)^{2^{3|\varphi|}} + 1$.
%
%
\end{theorem}

\begin{proof}
Let  $\cG=(\bbP_\bbD,\cL)$ be an $\bbN$-compass generator which features $\varphi$,
and let $y_{inf}\in D$ satisfy the conditions of Definition \ref{def:compassgennat}. 
We define a minimal set $S=\{\oy_0,\ldots,\oy_m\}$ of elements in $D$ 
such that  (i) $\oy_0 = 0$, (ii) $\oy_j<\oy_{j+1}$, for each $0\leq j < m$, 
(iii) $\oy_{m-1} = y_{inf}$, (iv) $\oy_m= y_{max}$, and (v) for every 
$(F,\Psi, h) \in \cA^{\cM}_\varphi$, if $\cC_{y_{inf}}(F,\Psi, h) >0$, then 
there exists $\oy_j$ such that $\cM(\oy_j,y_{inf}) = (F,\Psi, h)$. From the
minimality requirement, it follows that $m\leq  2^{3|\varphi|}+3$.

We build a finite sequence of $\bbN$-compass generators $\cG_0 \supset \cG_1 \supset 
\ldots \supset \cG_n$, whose last element is a small enough $\bbN$-compass generator
$\cG_n$, as follows.
We start with $\cG_0 = \cG$. Now, let $\cG_i=(\bbP_{\bbD_i}, \cL_i)$ be the $i$-th 
compass generator in the sequence, and let $S_i=\{\oy_0,\ldots,\oy_m\}$ be 
the above-defined minimal set of elements in $D_i$. 
If there exist no $y, y'$, with $\oy_j \leq y < y'<\oy_{j+1}$ for some $0 \leq j < m$, 
such that $\cC_{y}\equiv\cC_{y'}$, we terminate the construction and put $n = i$, that 
is, $\cG_i$ is the last $\bbN$-compass generator in the sequence. 
Otherwise, we must distinguish two cases. If $y_{inf} \leq y, y' < y_{max}$, then
the application of (the construction of) Lemma~\ref{lem:removal} to the pair of 
positions $y$ and $y'$ produces an $\bbN$-compass generator $\cG_{i+1} = 
(\bbP_{\bbD_{i+1}}, \cL_{i+1})$, with $\lenf{D_{i+1}} = \lenf{D_i} - (y'-y)$. 
It can be easily checked that the resulting structure satisfies the conditions 
of Definition \ref{def:compassgennat} (notice that some triples may disappear
from $y_{max}$, that is, $\cC_{y_{max}}(F,\Psi, h)$ may become equal to $0$
for some triple $(F,\Psi, h)$).
If $\oy_j \leq y, y' < \oy_{j+1}$ for some $j \leq \oy_{m-2}$, we can still
apply (the construction of) Lemma~\ref{lem:removal} to the pair of positions 
$y$ and $y'$, but we must guarantee that all triples belonging to the row $y_{inf}$
in $D_{i}$ are preserved. This can be done by an appropriate choice of the 
witnesses at step 4 of (the construction of) Lemma~\ref{lem:removal}. 
It is worth noticing that in both cases, while positions between $\oy_{j+1}$
and $\oy_{m-2}$ (if any) remain unchanged (they are only shifted), those between 
$\oy_{1}$ and $\oy_{j}$ may change from $S_i$ to $S_{i+1}$.

At the end of the procedure, all the horizontal configurations in between two consecutive
elements $\oy_j,\oy_{j+1} \in S$ are pairwise non-equivalent. From this, it immediately
follows that the final $\bbN$-compass generator $\cG_n=(\bbP_{\bbD_n},\cL_n)$ is such that 
$|D_n|\leq \left(2^{3|\varphi|}+2\right) \cdot \left( \frac{|\varphi|^2}{4} + \frac{|\varphi|}{2} 
+1 \right)^{2^{3|\varphi|}} + 1$.
\end{proof}

%% file: Decidabilityinteger.tex
\newcommand{\cut}[1]{}
\section{Decidability of MPNL over the integers}

\begin{figure}
\begin{tikzpicture}[scale=0.65]


\fill[color=blue!20] (1.5,1.5) -- (5,5) -- (0,5) --(0,1.5);
\fill[color=red!20] (0,0) -- (1.5,1.5) -- (0,1.5);
\draw[very thick] (0,0) -- (5,5) node[pos=0.1, above=0.1cm](T0) {$\cT_0$};
\draw[very thick] (0,5) -- (5,5);
\draw[very thick] (0,1.5) -- (0,0);
\draw[very thick, dashed] (0,1.5) -- (0,2);
\draw[very thick, dashed] (0,2.3) -- (0,5);
\draw[very thick] (0,1.5) -- (1.5,1.5);
\draw[very thick](3.5,3.5) -- (0,3.5);
\draw[very thick](3.5,3.5) -- (3.5,5);
\draw[very thick](1.5,1.5) -- (1.5,5);
\node[shape=circle,draw=black,inner sep=1.5pt,fill=black, 
](D) at (1.5,1.5) {};
\node[fill=blue!20,inner sep=1pt](DLAB)[above of=D, node distance=0.45cm] {$D$};
\node[shape=circle,draw=black,inner sep=1.5pt,fill=blue, label={[label distance=0.1cm]above:$C$}](C) at (1,1.5) {};
\node[shape=circle,draw=black,inner sep=1.5pt,fill=green, label={[label distance=0.1cm]above:$B$}](B) at (0.5,1.5) {};


\draw[very thick]  (-2.5,0) -- (-2.5,-0.3);


\draw[very thick] (-2.5,-2.5) -- (0,0);
\draw[very thick] (-2.5,0) -- (0,0);
\draw[very thick] (0,1.5) -- (-2.5,1.5);
\draw[very thick] (0,5) -- (-2.5,5);
\draw[very thick] (-2.5,3.5) -- (0,3.5);
\node[shape=circle,draw=black,inner sep=1.5pt,fill=black, label={[label distance=0.1cm]below:$D$}](Dp) at (0,0) {};
\node[shape=circle,draw=black,inner sep=1.5pt,fill=blue](Cp) at (-0.5,0) {};
\node[fill=white,inner sep=1pt](CpLAB)[below of=Cp, node distance=0.45cm]{$C$};
\node[shape=circle,draw=black,inner sep=1.5pt,fill=green, label={[label distance=0.1cm]below:$B$}](Bp) at (-1,0) {};
\node[shape=circle,draw=black,inner sep=1.5pt,fill=yellow, label={[label distance=0.1cm]below:$A$}](Ap1) at (-1.5,0) {};
\node[shape=circle,draw=black,inner sep=1.5pt,fill=yellow, label={[label distance=0.1cm]below:$A$}](Ap2) at (-2,0) {};
\draw[very thick]  (-2.5,0) -- (-2.5,-0.3);
\node[shape=circle,draw=black,inner sep=1.5pt,fill=yellow, label={[label distance=0.1cm]below:$A$}](Ap3) at (-2.5,0) {};
\draw[very thick]  (-2.5,-2.5) -- (-2.5,-0.9);

\node[shape=circle,draw=black,inner sep=1.5pt,fill=gray!50](aCp) at (-0.5,1.5) {};
\node[shape=circle,draw=black,inner sep=1.5pt,fill=gray!50](aBp) at (-1,1.5) {};
\node[shape=circle,draw=black,inner sep=1.5pt,fill=gray!50](aAp1) at (-1.5,1.5) {};
\node[shape=circle,draw=black,inner sep=1.5pt,fill=gray!50](aAp2) at (-2,1.5) {};
\node[shape=circle,draw=black,inner sep=1.5pt,fill=gray!50](aAp3) at (-2.5,1.5) {};
\draw[very thick](Ap1) -- (aAp1) node[fill=white,inner sep=1pt, pos=0.5](1){$3$};
\draw[very thick](Ap2) -- (aAp2) node[fill=white,inner sep=1pt, pos=0.5](1){$2$};
\draw[very thick](Bp) -- (aBp) node[fill=white,inner sep=1pt, pos=0.5](1){$4$};
\draw[very thick](Cp) -- (aCp) node[fill=white,inner sep=1pt, pos=0.5](1){$5$};
\draw[very thick](Ap3) -- (aAp3) node[fill=white,inner sep=1pt, pos=0.5](1){$1$};

\node[shape=circle,draw=black,inner sep=1.5pt,fill=red](aaCp) at (-0.5,5) {};
\node[shape=circle,draw=black,inner sep=1.5pt,fill=red](aaBp) at (-1,5) {};
\node[shape=circle,draw=black,inner sep=1.5pt,fill=red](aaAp1) at (-1.5,5) {};
\node[shape=circle,draw=black,inner sep=1.5pt,fill=red](aaAp2) at (-2,5) {};
\node[shape=circle,draw=black,inner sep=1.5pt,fill=red](aaAp3) at (-2.5,5) {};
\draw[very thick] (aAp3) -- (aaAp3) node[fill=white,inner sep=1pt, pos=0.4](1){$a$};
\draw[very thick] (aAp2) -- (aaAp2) node[fill=white,inner sep=1pt, pos=0.4](1){$b$};
\draw[very thick] (aAp1) -- (aaAp1) node[fill=white,inner sep=1pt, pos=0.4](1){$c$};
\draw[very thick] (aBp) -- (aaBp) node[fill=white,inner sep=1pt, pos=0.4](1){$d$};
\draw[very thick] (aCp) -- (aaCp) node[fill=white,inner sep=1pt, pos=0.4](1){$e$};

\node[shape=circle,draw=black,inner sep=1.5pt,fill=yellow, label={[label distance=0.1cm]above:$A$}](B) at (0,1.5) {};

\node(YMAX) at (-3.5,5) {$y_{max}$};
\node(YINF) at (-3.5,3.5) {$y_{fut}$};
\node(YFIN) at (-3.5,1.5) {$0$};
\node(YPAST) at (-3.5,0) {$y_{past}$};
\node(YMIN) at (-3.5,-2.5) {$y_{min}$};

\node(LABEL) at (-1,-2.5) {(a)};


\pgftransformshift{\pgfpoint{6cm}{0cm}}

\fill[color=blue!20] (1.5,1.5) -- (5,5) -- (0,5) --(0,1.5);
\fill[color=red!20] (0,0) -- (1.5,1.5) -- (0,1.5);
\draw[very thick] (0,0) -- (5,5) node[pos=0.1, above=0.1cm](T0) {$\cT_0$};
\draw[very thick] (0,5) -- (5,5);
\draw[very thick] (0,1.5) -- (0,0);
\draw[very thick, dashed] (0,1.5) -- (0,2);
\draw[very thick, dashed] (0,2.3) -- (0,5);
\draw[very thick] (0,1.5) -- (1.5,1.5);
\draw[very thick](3.5,3.5) -- (0,3.5);
\draw[very thick](3.5,3.5) -- (3.5,5);
\draw[very thick](1.5,1.5) -- (1.5,5);
\node[shape=circle,draw=black,inner sep=1.5pt,fill=black](D) at (1.5,1.5) {};
\node[fill=blue!20,inner sep=1pt](DLAB)[above of=D, node distance=0.45cm] {$D$};
\node[shape=circle,draw=black,inner sep=1.5pt,fill=blue, label={[label distance=0.1cm]above:$C$}](C) at (1,1.5) {};
\node[shape=circle,draw=black,inner sep=1.5pt,fill=green, label={[label distance=0.1cm]above:$B$}](B) at (0.5,1.5) {};
\node[shape=circle,draw=black,inner sep=1.5pt,fill=yellow, label={[label distance=0.1cm]above:$A$}](A) at (0,1.5) {};

\pgftransformshift{\pgfpoint{-1.5cm}{-1.5cm}}

\fill[color=red!20] (0,0) -- (1.5,1.5) -- (0,1.5);
\draw[very thick] (0,0) -- (1.5,1.5) node[pos=0.3, above=0.1cm](T0) {$\cT_1$};
\draw[very thick] (0,1.5) -- (0,0);
\draw[very thick] (0,1.5) -- (1.5,1.5);
\node[shape=circle,draw=black,inner sep=1.5pt,fill=black](aD) at (1.5,1.5) {};
\node[shape=circle,draw=black,inner sep=1.5pt,fill=blue](aC) at (1,1.5) {};
\node[shape=circle,draw=black,inner sep=1.5pt,fill=green](aB) at (0.5,1.5) {};
\node[shape=circle,draw=black,inner sep=1.5pt,fill=yellow](aA) at (0,1.5) {};

\pgftransformshift{\pgfpoint{-1.5cm}{-1.5cm}}
\cut{
\fill[color=red!20] (0,0) -- (1.5,1.5) -- (0,1.5);
\draw[very thick] (0,0) -- (1.5,1.5) node[pos=0.3, above=0.1cm](T0) {$\cT_2$};
\draw[very thick] (0,1.5) -- (0,0);
\draw[very thick] (0,1.5) -- (1.5,1.5);}
\node[shape=circle,draw=black,inner sep=1.5pt,fill=black](bD) at (1.5,1.5) {};

\cut{\node[shape=circle,draw=black,inner sep=1.5pt,fill=blue](bC) at (1,1.5) {};
\node[shape=circle,draw=black,inner sep=1.5pt,fill=green](bB) at (0.5,1.5) {};
\node[shape=circle,draw=black,inner sep=1.5pt,fill=yellow](bA) at (0,1.5) {};
}
\pgftransformshift{\pgfpoint{-1.5cm}{-1.5cm}}
\cut{
\fill[color=red!20] (0,0) -- (1.5,1.5) -- (0,1.5);
\draw[very thick] (0,0) -- (1.5,1.5) node[pos=0.3, above=0.1cm](T0) {$\cT_3$};
\draw[very thick] (0,1.5) -- (0,0);
\draw[very thick] (0,1.5) -- (1.5,1.5);
\node[shape=circle,draw=black,inner sep=1.5pt,fill=black](cD) at (1.5,1.5) {};
\node[shape=circle,draw=black,inner sep=1.5pt,fill=blue](cC) at (1,1.5) {};
\node[shape=circle,draw=black,inner sep=1.5pt,fill=green](cB) at (0.5,1.5) {};
\node[shape=circle,draw=black,inner sep=1.5pt,fill=yellow](cA) at (0,1.5) {};

\node[shape=circle,draw=black,inner sep=1.5pt,fill=gray!50](cA1) at (0,3){};
\node[shape=circle,draw=black,inner sep=1.5pt,fill=gray!50](cC1) at (1,3) {};
\node[shape=circle,draw=black,inner sep=1.5pt,fill=gray!50](cB1) at (0.5,3) {};

\draw[very thick](cA) -- (cA1) node[fill=white,inner sep=1pt, pos=0.5]{$1$};
\draw[very thick](cB) -- (cB1) node[fill=white,inner sep=1pt, pos=0.5]{$4$};
\draw[very thick](cC) -- (cC1) node[fill=white,inner sep=1pt, pos=0.5]{$5$};

\node[shape=circle,draw=black,inner sep=1.5pt,fill=gray!50](cA2) at (0,4.5){};
\node[shape=circle,draw=black,inner sep=1.5pt,fill=gray!50](cC2) at (1,4.5) {};
\node[shape=circle,draw=black,inner sep=1.5pt,fill=gray!50](cB2) at (0.5,4.5) {};
\node[shape=circle,draw=black,inner sep=1.5pt,fill=gray!50](bA1) at (1.5,4.5){};
\node[shape=circle,draw=black,inner sep=1.5pt,fill=gray!50](bC1) at (2.5,4.5) {};
\node[shape=circle,draw=black,inner sep=1.5pt,fill=gray!50](bB1) at (2,4.5) {};
}
\cut{
\draw[very thick](bA) -- (bA1) node[fill=white,inner sep=1pt, pos=0.5]{$2$};
\draw[very thick](bB) -- (bB1) node[fill=white,inner sep=1pt, pos=0.5]{$4$};
\draw[very thick](bC) -- (bC1) node[fill=white,inner sep=1pt, pos=0.5]{$5$};
\draw[very thick](cA1) -- (cA2) node[fill=white,inner sep=1pt, pos=0.5]{$1$};
\draw[very thick](cB1) -- (cB2) node[fill=white,inner sep=1pt, pos=0.5]{$4$};
\draw[very thick](cC1) -- (cC2) node[fill=white,inner sep=1pt, pos=0.5]{$4$};

\node[shape=circle,draw=black,inner sep=1.5pt,fill=gray!50](cA3) at (0,6){};
\node[shape=circle,draw=black,inner sep=1.5pt,fill=gray!50](cC3) at (1,6) {};
\node[shape=circle,draw=black,inner sep=1.5pt,fill=gray!50](cB3) at (0.5,6) {};
\node[shape=circle,draw=black,inner sep=1.5pt,fill=gray!50](bA2) at (1.5,6){};
\node[shape=circle,draw=black,inner sep=1.5pt,fill=gray!50](bC2) at (2.5,6) {};
\node[shape=circle,draw=black,inner sep=1.5pt,fill=gray!50](bB2) at (2,6) {};
}
\node[shape=circle,draw=black,inner sep=1.5pt,fill=gray!50](aC1) at (4,6) {};
\node[shape=circle,draw=black,inner sep=1.5pt,fill=gray!50](aB1) at (3.5,6) {};
\node[shape=circle,draw=black,inner sep=1.5pt,fill=gray!50](aA1) at (3,6) {};

\cut{
\draw[very thick](bA2) -- (bA1) node[fill=white,inner sep=1pt, pos=0.5]{$2$};
\draw[very thick](bB2) -- (bB1) node[fill=white,inner sep=1pt, pos=0.5]{$4$};
\draw[very thick](bC2) -- (bC1) node[fill=white,inner sep=1pt, pos=0.5]{$5$};
\draw[very thick](cA3) -- (cA2) node[fill=white,inner sep=1pt, pos=0.5]{$1$};
\draw[very thick](cB3) -- (cB2) node[fill=white,inner sep=1pt, pos=0.5]{$4$};
\draw[very thick](cC3) -- (cC2) node[fill=white,inner sep=1pt, pos=0.5]{$4$};
}

\draw[very thick](aA) -- (aA1) node[fill=white,inner sep=1pt, pos=0.5]{$2$};
\draw[very thick](aB) -- (aB1) node[fill=white,inner sep=1pt, pos=0.5]{$4$};
\draw[very thick](aC) -- (aC1) node[fill=white,inner sep=1pt, pos=0.5]{$5$};

\node(LABEL) at (3,2) {(b)};


\pgftransformshift{\pgfpoint{13.5cm}{4.5cm}}

\fill[color=blue!20] (1.5,1.5) -- (5,5) -- (0,5) --(0,1.5);
\fill[color=red!20] (0,0) -- (1.5,1.5) -- (0,1.5);
\draw[very thick] (0,0) -- (5,5) node[pos=0.1, above=0.1cm](T0) {$\cT_0$};
\draw[very thick] (0,5) -- (5,5);
\draw[very thick] (0,1.5) -- (0,0);
\draw[very thick, dashed] (0,1.5) -- (0,2);
\draw[very thick, dashed] (0,2.3) -- (0,5);
\draw[very thick] (0,1.5) -- (1.5,1.5);
\draw[very thick](3.5,3.5) -- (0,3.5);
\draw[very thick](3.5,3.5) -- (3.5,5);
\draw[very thick](1.5,1.5) -- (1.5,5);
\node[shape=circle,draw=black,inner sep=1.5pt,fill=black](D) at (1.5,1.5) {};
\node[fill=blue!20,inner sep=1pt](DLAB)[above of=D, node distance=0.45cm] {$D$};
\node[shape=circle,draw=black,inner sep=1.5pt,fill=blue, label={[label distance=0.1cm]above:$C$}](C) at (1,1.5) {};
\node[shape=circle,draw=black,inner sep=1.5pt,fill=green, label={[label distance=0.1cm]above:$B$}](B) at (0.5,1.5) {};
\node[shape=circle,draw=black,inner sep=1.5pt,fill=yellow, label={[label distance=0.1cm]above:$A$}](A) at (0,1.5) {};

\pgftransformshift{\pgfpoint{-1.5cm}{-1.5cm}}

\fill[color=red!20] (0,0) -- (1.5,1.5) -- (0,1.5);
\draw[very thick] (0,0) -- (1.5,1.5) node[pos=0.3, above=0.1cm](T0) {$\cT_1$};
\draw[very thick] (0,1.5) -- (0,0);
\draw[very thick] (0,1.5) -- (1.5,1.5);
\node[shape=circle,draw=black,inner sep=1.5pt,fill=black](aD) at (1.5,1.5) {};
\node[shape=circle,draw=black,inner sep=1.5pt,fill=blue](aC) at (1,1.5) {};
\node[shape=circle,draw=black,inner sep=1.5pt,fill=green](aB) at (0.5,1.5) {};
\node[shape=circle,draw=black,inner sep=1.5pt,fill=yellow](aA) at (0,1.5) {};

\pgftransformshift{\pgfpoint{-1.5cm}{-1.5cm}}

\fill[color=red!20] (0,0) -- (1.5,1.5) -- (0,1.5);
\draw[very thick] (0,0) -- (1.5,1.5) node[pos=0.3, above=0.1cm](T0) {$\cT_2$};
\draw[very thick] (0,1.5) -- (0,0);
\draw[very thick] (0,1.5) -- (1.5,1.5);
\node[shape=circle,draw=black,inner sep=1.5pt,fill=black](bD) at (1.5,1.5) {};
\node[shape=circle,draw=black,inner sep=1.5pt,fill=blue](bC) at (1,1.5) {};
\node[shape=circle,draw=black,inner sep=1.5pt,fill=green](bB) at (0.5,1.5) {};
\node[shape=circle,draw=black,inner sep=1.5pt,fill=yellow](bA) at (0,1.5) {};

\pgftransformshift{\pgfpoint{-1.5cm}{-1.5cm}}

\node[shape=circle,draw=black,inner sep=1.5pt,fill=black](cD) at (1.5,1.5) {};



\node[shape=circle,draw=black,inner sep=1.5pt,fill=gray!50](bA1) at (1.5,4.5){};
\node[shape=circle,draw=black,inner sep=1.5pt,fill=gray!50](bC1) at (2.5,4.5) {};
\node[shape=circle,draw=black,inner sep=1.5pt,fill=gray!50](bB1) at (2,4.5) {};

\draw[very thick](bA) -- (bA1) node[fill=white,inner sep=1pt, pos=0.5]{$2$};
\draw[very thick](bB) -- (bB1) node[fill=white,inner sep=1pt, pos=0.5]{$4$};
\draw[very thick](bC) -- (bC1) node[fill=white,inner sep=1pt, pos=0.5]{$5$};

\node[shape=circle,draw=black,inner sep=1.5pt,fill=gray!50](bA2) at (1.5,6){};
\node[shape=circle,draw=black,inner sep=1.5pt,fill=gray!50](bC2) at (2.5,6) {};
\node[shape=circle,draw=black,inner sep=1.5pt,fill=gray!50](bB2) at (2,6) {};
\node[shape=circle,draw=black,inner sep=1.5pt,fill=gray!50](aC1) at (4,6) {};
\node[shape=circle,draw=black,inner sep=1.5pt,fill=gray!50](aB1) at (3.5,6) {};
\node[shape=circle,draw=black,inner sep=1.5pt,fill=gray!50](aA1) at (3,6) {};

\draw[very thick](bA2) -- (bA1) node[fill=white,inner sep=1pt, pos=0.5]{$1$};
\draw[very thick](bB2) -- (bB1) node[fill=white,inner sep=1pt, pos=0.5]{$1$};

\draw[very thick](bC2) -- (bC1) node[fill=white,inner sep=1pt, pos=0.5]{$3$};
\draw[very thick](aA) -- (aA1) node[fill=white,inner sep=1pt, pos=0.5]{$2$};
\draw[very thick](aB) -- (aB1) node[fill=white,inner sep=1pt, pos=0.5]{$4$};
\draw[very thick](aC) -- (aC1) node[fill=white,inner sep=1pt, pos=0.5]{$5$};

\cut{
\node[shape=circle,draw=black,inner sep=1.5pt,fill=red](cA4) at (0,9.5){};
\node[shape=circle,draw=black,inner sep=1.5pt,fill=red](cC4) at (1,9.5) {};
\node[shape=circle,draw=black,inner sep=1.5pt,fill=red](cB4) at (0.5,9.5) {};
}

\node[shape=circle,draw=black,inner sep=1.5pt,fill=red](bA3) at (1.5,9.5){};
\node[shape=circle,draw=black,inner sep=1.5pt,fill=red](bC3) at (2.5,9.5) {};
\node[shape=circle,draw=black,inner sep=1.5pt,fill=red](bB3) at (2,9.5) {};
\node[shape=circle,draw=black,inner sep=1.5pt,fill=red](aC2) at (4,9.5) {};
\node[shape=circle,draw=black,inner sep=1.5pt,fill=red](aB2) at (3.5,9.5) {};
\node[shape=circle,draw=black,inner sep=1.5pt,fill=red](aA2) at (3,9.5) {};

\cut{ \draw[very thick](cA3) -- (cA4) node[fill=white,inner sep=1pt, pos=0.4]{$a$}; 
\draw[very thick](cC3) -- (cC4) node[fill=white,inner sep=1pt, pos=0.4]{$d$};
\draw[very thick](cB3) -- (cB4) node[fill=white,inner sep=1pt, pos=0.4]{$e$};}

\draw[very thick](bA2) -- (bA3) node[fill=white,inner sep=1pt, pos=0.4]{$a$};
\draw[very thick](bC2) -- (bC3) node[fill=white,inner sep=1pt, pos=0.4]{$b$};
\draw[very thick](bB2) -- (bB3) node[fill=white,inner sep=1pt, pos=0.4]{$a$};
\draw[very thick](aC1) -- (aC2) node[fill=white,inner sep=1pt, pos=0.4]{$e$};
\draw[very thick](aB1) -- (aB2) node[fill=white,inner sep=1pt, pos=0.4]{$d$};
\draw[very thick](aA1) -- (aA2) node[fill=white,inner sep=1pt, pos=0.4]{$c$};

\node(LABEL) at (3,1.8) {(c)};
\end{tikzpicture}
\label{fig:intproof}
\caption{From a $\bbZ$-compass generator to a compass structure over $\bbZ$.}
\end{figure}
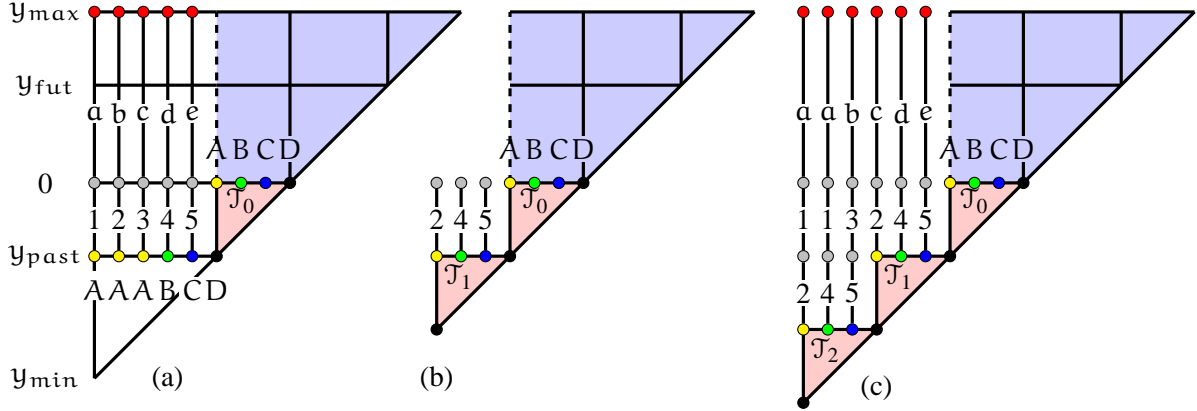

\newcommand{\ox}{\overline{x}}

In this section, we generalize the notion of \emph{compass generator} in order 
to prove the decidability of the satisfiability problem for MPNL over $\bbZ$.

\begin{definition}\label{def:compassgenint}
Let $\varphi$ be an MPNL formula. A \emph{$\bbZ$-compass generator} for $\varphi$ 
is a finite compass $\varphi$-structure $\cG=(\bbP_\bbD,\cL)$ such that there exist
$y_{fut}, y_{past} \in D$, with $y_{past} < 0 < y_{fut}$, $y_{past} - y_{min}
\geq k_\varphi$, and $y_{max} - y_{fut} \geq k_\varphi$, which satisfy the 
following conditions:
\begin{compactenum}
\item all $\dl$-requests of every point $(y,y)\in\bbP_\bbD$, with $y_{past}
\leq y \leq y_{max}$, are fulfilled;
\item for every $(F,\Psi, h) \in \cA^{\cM}_\varphi$, if $\cC_{y_{max}}(F,\Psi, h) >0$,
			then $\cC_{y_{fut}}(F,\Psi, h) >0$, and $\cM(x,y_{max})=(F,\emptyset, h)$, 
			for every $y_{min} \leq x \leq y_{fut}$;
\item for every $(F,\Psi, h) \in \cA^{\cM}_\varphi$, if $\cC_{y_{past}}(F,\Psi, h) >0$,
	then there exists $y_{past} \leq x \leq 0$ such that $\cM(x,0)=(F,\Psi, h)$.			
			
\end{compactenum}
\end{definition}

\begin{theorem}\label{teo:satiffcompassgenint}
An MPNL formula $\varphi$ is satisfiable over $\bbZ$ iff there exists
a $\bbZ$-compass generator for it.
\end{theorem}

\begin{proof} 
We start with the left-to-right direction. 
From the satisfiability of $\varphi$ over $\bbZ$, it follows that 
there exists a fulfilling compass $\varphi$-structure $\cG=\langle\bbP_{\bbZ}, 
\cL\rangle$ which features $\varphi$.
Hence, to prove the claim it suffices to show that there exist $y_{min}, y_{past}, 
y_{fut}$, and $y_{max}$, with $y_{min} < y_{past} < 0 < y_{fut} < y_{max}$, that 
satisfy the conditions of Definition \ref{def:compassgenint}. 
Since the index of $\equiv$ is finite, there exists an infinite-to-the-past 
sequence of elements $\cS=y_{-1}>y_{-2}>\ldots$ such 
that, for every $i,j\in \bbN$, $\cC_{y_i} \equiv\cC_{y_j}$.  
Without loss of generality, we can assume that $y_{-1}=0$.
Since $\cS$ is infinite to the past, there exists $j<-1$ 
such that, for every $(F,\Psi, h) \in \cA^{\cM}_\varphi$ with 
$\cC_{y_j}(F,\Psi, h)>0$, there exists $y_j\leq x\leq y_{-1}$, 
with $\cM(x,y_{-1})=(F,\Psi,h)$. We put $y_{past}=y_{j}$. 
The elements $y_{max}$ and $y_{fut}$ can be selected using 
the very same argument of the proof of Theorem \ref{teo:satiffcompassgen} 
guaranteeing that $0 < y_{fut} <y_{max}$. 
Next, we take an element $\oy < y_{past}$ such that, for every 
$y_{past}\leq y \leq y_{max}$ and  every $\psi \in \req_l(\cL(y,y))$,
there exists an element $\oy\leq x\leq y$ such that $\psi \in \cL(x,y)$.
We put $y_{min}=\oy$.
Let $\cG'=\langle\bbP_{\bbD}, \cL'\rangle$ be a compass $\varphi$-structure
such that $D=\{y_{min},\ldots,y_{max}\}$ and, for every $(x,y)\in\bbP_{\bbD}$, 
$\cL'(x,y)=\cL(x,y)$. It can be easily checked that $\cG'$ is a 
$\bbZ$-compass generator for $\varphi$. 

\medskip

The right-to-left direction is much more involved with respect to the case of
$\bbN$. We give a sketch of the proof only, making use of the 
pictorial representation given in Figure \ref{fig:intproof}.
Figure \ref{fig:intproof}.a depicts a $\bbZ$-compass generator 
$\cG=\langle\bbP_{\bbD}, \cL\rangle$ for some MPNL formula 
$\varphi$. The vertical segments that will be used to fill in the
gaps that will appear during the construction of the infinite prefix
are suitably numbered; lowercase letters will be used to identify
the vertical segments that will be exploited to fill in the gaps 
in between $0$ and $y_{max}$; upper case letters identify 
the marked atoms. 

We first define the labeling of points $(x,y)$, 
with $x \leq y \leq y_{max}$ (the infinite prefix). To this end, 
we leave the labeling of points $(x,y)$, with $y_{past} \leq x 
\leq y \leq y_{max}$, unchanged, and we define the 
labeling of the other points as follows (in particular, we 
suitably redefine the labeling of points $(x,y)$, with $y_{min} \leq x 
< y_{past}$ and $y_{min} \leq y \leq y_{max}$). 

Let us first consider the $\dl$-requests of points $(x,x)$, 
with $y_{past}\leq x\leq 0$. By condition 1 of Definition 
\ref{def:compassgenint}, all of them are satisfied in $\cG$.
We rearrange the structure of $\cG$ in order to generate a 
fulfilling infinite-to-the-past compass $\varphi$-structure $\cG'$. 
To give an intuitive account of the construction, suppose that the 
set of points that satisfy the $\dl$-requests is included in the 
set of points belonging to the vertical segments $1, \ldots, 5$ of Figure 
\ref{fig:intproof}.a. By exploiting condition 3 of Definition 
\ref{def:compassgenint}, we generate a sufficient number of
copies $\cT_1, \ldots, \cT_n$ of the triangle $\cT_0$ ($2$ 
copies in Figure \ref{fig:intproof}.b), and we append them 
one below the other starting from $\cT_0$ ($\cT_1$ immediateley
below $\cT_0$, $\cT_2$ immediately below $\cT_1$, and so on).
$\dl$-requests involving length constraints with $k < k_\varphi$ 
are satisfied by points belonging to the vertical segments rooted 
at the right end of the horizontal edge of $\cT_1$ only (segments 
$4$, and $5$); the other $\dl$-requests are satisfied by points 
belonging to the vertical segments rooted at the left end of 
the horizontal edge of $\cT_1$, at $\cT_2, \ldots$, and at $\cT_n$ 
(segments $1, 2,$ and $3$). Notice that vertical segments in $\cG$ which are
sufficiently far way from the diagonal (points $(x,y)$ such that
$y - x \geq \kphi$) are insensitive to $\dep{R}$-preserving 
changes of the labels of their endpoints (segments $1, 2,$ and 
$3$ in Figure \ref{fig:intproof}.c). 

Let us consider now points $(x,x)$, with $0\leq x\leq y_{max}$, and
suppose that the set of points that satisfy their $\dl$-requests in $\cG$
is included in the set of points belonging to the vertical segments $a, 
b, c, d,$ and $e$ of Figure \ref{fig:intproof}.a. In $\cG'$, these 
$\dl$-requests are satisfied by (re)introducing the vertical segments 
$a, b, c, d,$ and $e$ above the appropriate vertical segments $1, 2, 
3, 4,$ and $5$, possibly duplicating some of them (this is the case with $a$ 
in Figure \ref{fig:intproof}.c). As before, vertical segments in 
$\cG$ which are sufficiently far way from the diagonal (points $(x,y)$ 
such that $y - x \geq \kphi$) are insensitive to $\dep{R}$-preserving 
changes of the labels of their endpoints (segments $b$ and $c$
in Figure \ref{fig:intproof}.c). 

The procedure that we applied to fulfill the $\dl$-requests of points 
$(x,x)$, with $y_{past}\leq x\leq 0$, can then be applied to satisfy the 
$\dl$-requests of points $(x,x)$, with $2 \cdot y_{past}\leq x
\leq y_{past}$, of points $(x,x)$, with $3 \cdot y_{past}\leq x
\leq 2 \cdot y_{past}$, and so on, to obtain a correct
labeling for all points $(x,x)$ of the infinite prefix.
 
To complete the labeling of $\cG'$,
we need to specify the labeling of points $(x,y)$, with $y_{max} < 
y$ (the infinite suffix). To this end, we apply the procedure of Theorem 
\ref{teo:satiffcompassgen} to $y_{fut}$ and $y_{max}$. The resulting 
compass $\varphi$-structure $\cG'= \langle\bbP_{\bbZ},\cL'\rangle$ is 
a fulfilling compass $\varphi$-structure featuring $\varphi$.
\end{proof}

\begin{theorem}\label{teo:boundcompassgenint}
Let $\varphi$ be an MPNL formula. If there exists a  $\bbZ$-compass 
generator $\cG=(\bbP_\bbD,\cL)$ that features $\varphi$, then there 
exists a $\bbZ$-compass generator $\cG'=(\bbP_{\bbD'},\cL')$, that 
features $\varphi$, with $|D'|\leq \left( 2^{3|\varphi|+1}+4\right) 
\cdot \left( \frac{|\varphi|^2}{4} + \frac{|\varphi|}{2} +
1 \right)^{2^{3|\varphi|}} + 1$.
\end{theorem}

\begin{proof}

Let  $\cG=(\bbP_\bbD,\cL)$ be a $\bbZ$-compass generator, that features 
$\varphi$, and let $y_{fut} $ and $y_{past}\in D$ satisfy the conditions 
of Definition \ref{def:compassgenint}. 
We define a minimal set $S=\{\oy_0,\ldots,\oy_m\}$ of elements in $D$ 
such that (i) $\oy_0 = y_{min}$, (ii) $\oy_{m}=y_{max}$, (iii) $\oy_{m-1}=y_{fut}$,
(iv) $\oy_j<\oy_{j+1}$, for each $0 \leq j < m$, (v) $\oy_j= 0$, for some 
$1 < j < m$, (vi) $\oy_{j'} = y_{past}$, for some $1 < j' < j$, (vii) 
$\oy_{j''} = y_{fut}$, for some $j < j'' < m$, (viii) for every $(F,\Psi, h) 
\in \cA^{\cM}_\varphi$, if $ \cC_{y_{past}}((F,\Psi, h)) > 0$, then there 
exists $l \leq j$ such that $\cM(\oy_l, 0)=(F,\Psi, h)$, and (ix) for every 
$(F,\Psi, h) \in \cA^{\cM}_\varphi$, if $\cC_{y_{fut}}(F,\Psi, h) >0$, then 
there exists $\oy_l\leq\oy_{fut}$ such that $\cM(\oy_l,y_{fut}) = (F,\Psi, h)$. 
From the minimality requirement, it follows that $m\leq  2^{3|\varphi|+1}+5$.

We build a finite sequence of $\bbZ$-compass generators $\cG_0 \supset \cG_1 \supset 
\ldots \supset \cG_n$, whose last element is a small enough $\bbZ$-compass generator
$\cG_n$.
We start with $\cG_0 = \cG$. Now, let $\cG_i=(\bbP_{\bbD_i}, \cL_i)$ be the $i$-th 
compass generator in the sequence and let $S_i=\{\oy_0,\ldots,\oy_m\}$ be 
the above-defined minimal set of elements in $D_i$. If there exist no $y, y'$, with 
$\oy_j \leq y < y'<\oy_{j+1}$ for some $0 \leq j < m$, such that $\cC_{y}\equiv\cC_{y'}$,
we put $n = i$, and we end the construction. 
Otherwise, as in Theorem \ref{teo:boundcompassgen}, we apply (the construction of) 
Lemma~\ref{lem:removal} to $y$ and $y'$ to obtain a compass generator $\cG_{i+1} 
= (\bbP_{\bbD_{i+1}}, \cL_{i+1})$, with $\lenf{D_{i+1}} = \lenf{D_i} - (y'-y)$. 

At the end of the procedure, all the horizontal configurations in between two 
consecutive elements $\oy_j,\oy_{j+1} \in S$ are pairwise non-equivalent. From 
this, it immediately follows that the final $\bbZ$-compass generator $\cG_n=(\bbP_{\bbD_n},\cL_n)$ 
is such that $|D_n|\leq \left( 2^{3|\varphi|+1}+4\right) \cdot \left( \frac{|\varphi|^2}{4} + 
\frac{|\varphi|}{2} +1 \right)^{2^{3|\varphi|}} + 1$.
\end{proof}

%% file: Decisionprocedure.tex
\newcommand{\ocC}{\overline{\cC}}
\section{An EXPSPACE decision procedure}\label{sec:decisionproc}

In this section, we describe a decision procedure that solves the 
satisfiability problem for MPNL over the integer numbers. 
Both the procedure for the finite case and that for the
natural numbers can be easily tailored from it. Let $\varphi$
be the MPNL formula to check for satisfiability. In order to
establish whether or not there exists a finite model satisfying
$\varphi$, we can proceed as follows. First, we represent a finite 
model in $\bbZ$ by means of the following formula:

\[ \begin{gathered}\psi_{fin}=\#_{all} \wedge \# \wedge
\Dr \Dl (\neg \pi \wedge \Dl \neg \#_{all}) \wedge
\Dl \Dr (\neg \pi \wedge \Dr \neg \#_{all}) \wedge \\
\Dr (\neg \pi \wedge \Dr \Dl \neg \#_{all}) \wedge
\Dl (\neg \pi \wedge \Dl \Dr \neg \#_{all}) \wedge \\
\G ( \# \leftrightarrow (\#_{all} \vee \dr\dl 
\#_{all} \vee \dl\dr \#_{all} \vee (\dr\dr\dl \#_{all} \wedge 
\dl\dl\dr \#_{all}))), \end{gathered}\]
where $\G$ is the commonly-used universal modality \cite{thesispietro},
$\#_{all}$ holds over one and one interval that collects all
points of the finite model and $\#$ holds over all and only the 
subintervals of such a $\#_{all}$-interval.

Under the assumption that $\#_{all}$ and $\#$ do not appear in 
$\varphi$, we can replace $\varphi$ by a formula $tr(\varphi)$
such that $\varphi$ has a finite model if and only if 
$\psi_{fin} \wedge \dl \dr \dr tr(\varphi)$ has a model in $\bbZ$. 
The formula $tr(\varphi)$ is inductively defined as follows:
(i) if $\varphi = p$ or $\varphi=len_{<k}$, then $tr(\varphi)=\varphi\wedge \#$,
(ii) if $\varphi=\neg \psi$, then $tr(\varphi)=\neg \# \vee \neg tr(\psi)$,
(iii) if $\varphi=\psi_1 \vee \psi_2$, then $tr(\varphi)=(\psi_1\wedge \#) \vee (\psi_2 \wedge \#)$,
(iv) if $\varphi= \dr \psi$, then $tr(\varphi)=\dr(\# \wedge \psi)$,
(v) if $\varphi= \dl \psi$, then $tr(\varphi)=\dl(\# \wedge \psi)$. 

Similarly, it is possible to prove that an MPNL formula $\varphi$ 
has a model over the linear order of natural numbers if and only if 
$\psi_{nat} \wedge \dl \dr \dr tr(\varphi)$ has a model in $\bbZ$,
where  $tr(\varphi)$ is defined as above and $\psi_{nat}$ is defined
as follows:
\[ \begin{gathered}\psi_{nat}= \# \wedge \Dl \neg \# \wedge \Dl\Dl \neg \# \wedge \G((\neg 
\# \wedge \dr \#)\rightarrow(\Dr \# \wedge \Dr\Dr \#))\end{gathered}\]

\begin{figure}[!]
\scriptsize
\centering
\begin{code}{79mm}

\FUNCTION{GuessConfiguration}{}
\BEGIN
\FORALL (F,\Psi,h)\in \cA^\cM_\varphi, \ \ \  \cC(F,\Psi,h)\gets 0; \\
\LET S_r\subseteq\{ \psi \in \closure(\varphi\ |\ \dr\psi \in \closure(\varphi))\};\\
\LET S_{l}\subseteq \{ \psi \in \closure(\varphi\ |\ \dl\psi \in \closure(\varphi))\};\\
\FORALL 1\leq i < \kphi \\
\BEGIN
\LET F \text{ an atom s.t. } \req_r(F)=S_r \text{ and } Len(F)=i ;\\ 
\LET \Psi\subseteq \{ \psi \in \closure(\varphi)\ |\ \dr\psi \in \closure(\varphi))\} ;\\
\ocC(F,\Psi,i)\GETS 1;\\
\END\\
\FORALL (F,\Psi,\kphi)\in \cA^\cM_\varphi \text{ s.t. }   \req_r(F)=S_r
\\
\BEGIN
\LET 0\leq i\leq \kphi, \ \ \ \ocC(F,\Psi,h)\GETS i
\END\\
\RETURN\ \ocC;
\END \\ \\

\FUNCTION{Merge}{\cC,\cC'}
\BEGIN
\FORALL (F,\Psi,h)\in\cA^{\cM}_{\varphi}\\
\BEGIN
 \ocC(F,\Psi,h)\GETS \cC(F,\Psi,h)+\cC'(F,\Psi,h);\\
 \END\\
\RETURN\ \ocC;
\END
\\ \\
\FUNCTION{Len}{F}
\BEGIN
\IF \exists 1\leq h<\kphi \text{ s.t. } \neg len_{<h} \in F \wedge len_{<h+1} \in F \THEN \RETURN\ h
\ELSE \RETURN \kphi
\END
\\ \\

\end{code}
\begin{code}{71mm}

\FUNCTION{MA$\_$set}{\cC}
\BEGIN

R=\{ (F,\Psi,h)\ |\ \cC(F,\Psi,h)>0 \};\\
\RETURN\ R;

\END \\ \\

\FUNCTION{NC$\_$ZeroToFut}{\cC^{current}}
\BEGIN

\LET S_r\subseteq\{ \psi \in \closure(\varphi\ |\ \dr\psi \in \closure(\varphi))\};\\
\LET S_l\subseteq \{ \psi \in \closure(\varphi\ |\ \dl\psi \in \closure(\varphi))\};\\
\LET \begin{array}{l}F_\pi \text{ an atom with $len_{<1} \in F_\pi$, $\req_r(F_\pi)=S_r$,}\\\text{ and  $\req_l(F_\pi)=S_l$ };\end{array}\\
\FORALL (F,\Psi,h)\in \cA^\cM_\varphi, \cC(F,\Psi,h)\gets 0;\\

\cC(F_\pi,\req_r(F)\setminus F_\pi,1)\gets 1;\\
\FORALL (F,\Psi,h)\in\cA^{\cM}_{\varphi} \\
\BEGIN
\FOR(1\leq i\leq \cC^{current}(G, \Psi,h) )\\
\BEGIN
\IF h=\kphi \THEN k\GETS \kphi
\ELSE k\GETS h+1\\
\LET\begin{array}{l} G \text{ s.t. $Len(G)=k$, $\req_r(G)=S_r$,}\\\text{ and  $\req_l(G)=\req_l(F)$};\end{array}\\
\cC(G,\Psi\setminus G, k)\gets \cC(G,\Psi\setminus G, k) +1;\\
\END
\END\\
\IF \left(\begin{array}{l}\exists \psi \in S_l \text{ s. t. } \forall (F,\Psi,h)\in \cA^\cM_\varphi \text{ with } \psi \in A \\
\text{ we have } \cC(F,\Psi,h)=0\end{array} \right) \THEN \RETURN \FALSE;\\
\RETURN\ \cC;\\
\END\\ \\

\end{code}
\label{fig:codeaux}
\caption{Auxiliary procedures for checking the satisfiability of $\phi$ over the integers.}
\end{figure}

\begin{figure}[!]
\scriptsize
\centering
\begin{code}{79mm}

\FUNCTION{NC$\_$MinToPast}{\cC^{current}}
\BEGIN

\LET S_r\subseteq\{ \psi \in \closure(\varphi\ |\ \dr\psi \in \closure(\varphi))\};\\
\LET S_l\subseteq \{ \psi \in \closure(\varphi\ |\ \dl\psi \in \closure(\varphi))\};\\
\LET\begin{array}{l} F_\pi \text{ an atom with $len_{<1} \in F_\pi$, $\req_r(F_\pi)=S_r$},\\ \text{and  $\req_l(F_\pi)=S_l$ };\end{array}\\
\FORALL (F,\Psi,h)\in \cA^\cM_\varphi \cC(F,\Psi,h)\gets 0;\\
\cC(F_\pi,\req_r(F)\setminus F_\pi,1)\gets 1;\\
\FORALL (F,\Psi,h)\in\cA^{\cM}_{\varphi} \\
\BEGIN
\FOR(1\leq i\leq \cC^{current}(G, \Psi,h) )\\
\BEGIN
\IF h=\kphi \THEN k\GETS \kphi
\ELSE k\GETS h+1\\
\LET \begin{array}{l}G \text{ s.t. $Len(G)=k$, $\req_r(G)=S_r$},\\ \text{ and  $\req_l(G)=\req_l(F)$};\end{array}\\
\cC(G,\Psi\setminus G, k)\gets \cC'(G,\Psi\setminus G, k) +1;\\
\END
\END\\
\RETURN\ \cC;\\
\END
\\ \\
\FUNCTION{NC$\_$LeftRight}{\cC^{left},\cC^{right}}
\BEGIN

\LET S_r\subseteq\{ \psi \in \closure(\varphi\ |\ \dr\psi \in \closure(\varphi))\};\\
\LET S_l\subseteq \{ \psi \in \closure(\varphi\ |\ \dl\psi \in \closure(\varphi))\};\\
\LET \begin{array}{l}F_\pi \text{ an atom with $len_{<1} \in F_\pi$, $\req_r(F_\pi)=S_r$},\\ \text{ and  $\req_l(F_\pi)=S_l$ }\end{array};\\
\FORALL (F,\Psi,h)\in \cA^\cM_\varphi \ocC^{right}(F,\Psi,h)\gets 0;\\
\FORALL (F,\Psi,h)\in \cA^\cM_\varphi \ocC^{left}(F,\Psi,h)\gets 0;\\
\ocC^{right}(F_\pi,\req_r(F)\setminus F_\pi,1)\gets 1;\\
\FORALL (F,\Psi,h)\in\cA^{\cM}_{\varphi} \\
\BEGIN
\FOR(1\leq i\leq \cC^{right}(G, \Psi,h) )\\
\BEGIN
\IF h=\kphi \THEN k\GETS \kphi
\ELSE k\GETS h+1\\
\LET \begin{array}{l}G \text{ s.t. $Len(G)=k$, $\req_r(G)=S_r$,}\\ \text{ and  $\req_l(G)=\req_l(F)$};\end{array}\\
\ocC^{right}(G,\Psi\setminus G, k)\gets \ocC^{right}(G,\Psi\setminus G, k) +1;\\
\END\\
\FOR(1\leq i\leq \cC^{left}(G, \Psi,h) )\\
\BEGIN
\IF h=\kphi \THEN k\GETS \kphi
\ELSE k\GETS h+1\\
\LET\begin{array}{l} G \text{ s.t. $Len(G)=k$, $\req_r(G)=S_r$,}\\ \text{ and  $\req_l(G)=\req_l(F)$};\end{array}\\
\ocC^{left}(G,\Psi\setminus G, k)\gets \ocC^{left}(G,\Psi\setminus G, k) +1;\\
\END
\END\\
\IF \left( \begin{array}{l}\exists \psi \in S_l \text{ s. t. } \forall (F,\Psi,h)\in \cA^\cM_\varphi \text{ with } \psi \in A \\
\text{ we have } \ocC^{left}(F,\Psi,h)=\ocC^{right}(F,\Psi,h)=0 \end{array}\right) \THEN \RETURN \FALSE;\\
\RETURN\ (\ocC^{left},\ocC^{right});\\
\END

\end{code}
\begin{code}{71mm}

\FUNCTION{MPNL-INTEGER-SAT}{\varphi}
\BEGIN

BOUND\gets \left( 2^{3|\varphi|+1}+4\right) \cdot \left( \frac{|\varphi|^2}{4} + 
\frac{|\varphi|}{2} +1 \right)^{2^{3|\varphi|}} ;\\

\LET S_r\subseteq\{ \psi \in \closure(\varphi\ |\ \dr\psi \in \closure(\varphi))\};\\
\LET S_l\subseteq \{ \psi \in \closure(\varphi\ |\ \dl\psi \in \closure(\varphi))\};\\
\LET \begin{array}{l}F_\pi \text{ an atom with $len_{<1} \in F_\pi$, $\req_r(F_\pi)=S_r$,}\\\text{ and  $\req_l(F_\pi)=S_l$ };\end{array}\\
\FORALL (F,\Psi,h)\in \cA^\cM_\varphi, \ \ \cC^{min}(F,\Psi,h)\gets 0;\\
\cC^{min}(F_\pi,\req_r(F)\setminus F_\pi,1)\gets 1;\\
\cC^{past}\GETS GuessConfiguration();\\
\cC\GETS \cC^{min};\\
steps\GETS 0;\\
\WHILE(\cC \not\equiv \cC^{past} \vee steps<\kphi)\\
\BEGIN
\IF steps > BOUND \THEN \RETURN \FALSE\\
\cC\GETS NC\_MinToPast(\cC);\\
steps \GETS steps+1;\\
\END\\
\cC^{left}\GETS \cC;\\
\FORALL (F,\Psi,h)\in \cA^\cM_\varphi, \  \cC^{right}(F,\Psi,h)\gets 0;\\
steps\GETS 0;\\
\WHILE(MA\_set (\cC^{past}) \not\subseteq MA\_set(\cC^{right}))\\
\BEGIN
\IF steps > BOUND \THEN \RETURN \FALSE\\
(\cC^{left},\cC^{right})\GETS NC\_LeftRight(\cC^{left},\cC^{right});\\
steps \GETS steps+1;\\
\END\\
\cC^{fut}\GETS GuessConfiguration();\\
\cC\GETS Merge(\cC^{left},\cC^{right});\\
steps\GETS 0;\\
\WHILE(\cC\not\equiv \cC^{fut})\\
\BEGIN
\IF steps > BOUND \THEN \RETURN \FALSE\\
\cC \GETS NC\_ZeroToFut(\cC);\\
steps \GETS steps+1;\\
\END\\
\cC^{max}\GETS \cC;\\
\cC^{left}\GETS \cC;\\
\FORALL (F,\Psi,h)\in \cA^\cM_\varphi,\ \  \cC^{right}(F,\Psi,h)\gets 0;\\
steps\GETS 0;\\
\WHILE \left( \begin{array}{l} (MA\_set(Merge(\cC^{left},\cC^{right}))\\ \supseteq MA\_set(\cC^{max}) \rightarrow  \exists (F,\Psi,h)\in \cA^\cM_\varphi \\ \text{ with } \cC^{left}(F,\Psi,h)>0 \wedge \Psi\neq \emptyset)\\ \vee steps>\kphi  \end{array} \right)\\
\BEGIN
\IF steps > BOUND \THEN \RETURN \FALSE\\
(\cC^{left},\cC^{right})\GETS NC\_LeftRight(\cC^{left},\cC^{right});\\
steps \GETS steps+1;\\
\END\\
\RETURN\ \TRUE;

\END

\end{code}
\label{fig:code}
\caption{The procedure for checking the satisfiability of $\phi$ over the integers.}
\end{figure}
The detailed code of the decision procedure is reported in Figure \ref{fig:code}. 
It builds a tentative $\bbZ$-compass generator for $\varphi$ starting from $y_{min}$ and exploring 
two consecutive horizontal configurations at every step.
Every configuration is represented using an exponential number of counters, 
bounded by the maximum size for a $\bbZ$-compass generator given in Theorem 
\ref{teo:boundcompassgenint} (doubly exponential in the size of $|\varphi|$).
However, assuming that the values of all counters are encoded in binary, the 
maximum value for each counter takes an exponential storage space. The very same 
argument can be used to provide an exponential space bound for the $steps$ counter.
Moreover, the procedure needs to keep track of a constant number of horizontal 
configurations only ($\cC^{{min}}, \cC^{{past}}, \cC^{0}, \cC^{fut}, \cC^{{max}},\ocC, 
\cC,\cC',$ $\cC^{right}$, and $\cC^{left}$). Pairing this result with the 
EXPSPACE-hardness given in \cite{rpnl_with_constraints}, we can state the 
following theorem.

\begin{theorem}\label{teo:expspace}
The satisfiability problem for MPNL,  interpreted over (any subsets of) the integers 
is EXP-SPACE-complete.
\end{theorem}